\documentclass[a4paper,11pt]{scrartcl}

\usepackage[T1]{fontenc}
\usepackage{hyperref}
\usepackage{amsmath,amssymb,amsthm}
\usepackage{thmtools}
\usepackage{algorithm}
\usepackage{algpseudocode}

\usepackage{float}
\usepackage{graphicx}
\usepackage{enumerate}
\usepackage{todonotes}
\usepackage{paralist}
\usepackage{float}
\usepackage{listings}

\usepackage{multicol}

\usepackage{complexity}

\def\calF{\mathcal{F}}
\def\calD{\mathcal{D}}
\def\calX{\mathcal{X}}

\usepackage[font=small,labelfont=bf]{caption}
\usepackage[font=small,labelfont=normalfont,labelformat=simple]{subcaption}

\newcommand{\manfred}[1]{\textcolor{blue}{Manfred: #1}}

\DeclareMathOperator{\Inv}{Inv}

\newcommand{\True}{\textsc{true}\xspace}
\newcommand{\False}{\textsc{false}\xspace}
 
\newcommand{\Complete}{\textsc{Comp}\xspace}

\graphicspath{{figs_color/}}

\usepackage{xspace}

\def\NUMPROOFS{$41$\xspace}
\def\NUMSETTINGS{$144$\xspace}
\def\NUMGREEDY{$20$\xspace}

\def\NUMPROOFSFOUR{$2328$\xspace} 
\def\NUMSETTINGSFOURSELECTION{$4351$\xspace}

\newtheorem{theorem}{Theorem}
\numberwithin{theorem}{section} 
\newtheorem{lemma}[theorem]{Lemma}

\newtheorem{claim}{Claim}
\newtheorem{problem}{Problem}

\def\inst#1{$^{#1}$}

\usepackage{cleveref}

\begin{document}

\title{Finding hardness reductions automatically using SAT solvers\footnote{%
This work was initiated at the \emph{Order \& Geometry Workshop} in Ci\k{a}\.{z}e\'{n} in September~2022. We thank the organizers and all participants for the inspiring atmosphere. In particular, we thank Konrad Majewski for many fruitful discussions.
H.~Bergold was supported by the DFG-Research Training Group 'Facets of Complexity' (DFG-GRK~2434). 
M.~Scheucher was supported by the DFG Grant SCHE~2214/1-1. 
F. Schröder was supported by the GA\v{C}R Grant no. 23-04949X.}}

\author{
Helena Bergold\inst{1}
\and
Manfred Scheucher\inst{2}
\and
Felix Schröder\inst{2}\inst{3}
}

\date{\vspace{-1em}}

\maketitle
\begin{center}
{\footnotesize
\inst{1} 
Department of Computer Science, 
Freie Universit\"at Berlin, Germany\\
\texttt{firstname.lastname@fu-berlin.de}
\\\ \\
\inst{2} 
Institut f\"ur Mathematik, 
Technische Universit\"at Berlin, Germany\\
\texttt{lastname@math.tu-berlin.de}
\\\ \\
\inst{3} 
Department of Applied Mathematics, Faculty of Mathematics and Physics,\\
Charles University, Prague, Czech Republic\\
\texttt{schroder@kam.mff.cuni.cz}
\\\ \\
}
\end{center}

\begin{abstract}
In this article, we show that the completion problem, i.e. the decision problem whether a partial structure can be completed to a full structure, is \NP-complete for many combinatorial structures.  
While the gadgets for most reductions in literature are found by hand, we present an algorithm to construct gadgets in a fully automated way. 
Using our framework which is based on SAT,
we present the first thorough study of the completion problem on sign mappings with forbidden substructures by classifying thousands of structures for which the completion problem is \NP-complete.
Our list in particular includes interior triple systems, which were introduced by Knuth towards an axiomatization of planar point configurations. 
Last but not least, we give 
an infinite family of structures
generalizing interior triple system to higher dimensions 
for which the completion problem is \NP-complete.
\end{abstract}

\section{Introduction}
\label{sec:intro}

A fundamental question in combinatorics
which naturally appears in various contexts is whether a partial structure can be completed.
In certain cases, the
completion problem (short: \Complete) can be decided in polynomial time. Planar graphs, for example, can be completed to a triangulation by greedily adding edges.
However, 
\Complete is known to be \NP-hard for plenty of structures such as Latin squares \cite{Colbourn1984},
Sudokus and related logic puzzles \cite{YatoSeta2003},
Steiner triple systems \cite{Colbourn1983},
and many others~\cite{KantBodlaender1997,Bouchitte1987}.
Moreover, even though completability might be polynomial on a structure such as planar graphs or partial orders,
this can change drastically when restricting to a subclass such as
biconnected planar graphs with bounded maximum degree \cite{KantBodlaender1997}
or 
partial orders with a bounded number of jumps \cite{Bouchitte1987}.
Besides purely combinatorial structures, computationally hard completion problems also appear in combinatorial geometry, for example in simple topological drawings~\cite{AKPSVW2022}.

In this article, we consider the completion problem of various combinatorial structures and present an algorithm to find gadgets for a hardness reduction from 3SAT in an automatic way.
Even though the gadget-finding problem lies in the third level of the polynomial hierarchy,
our Python framework performs quite well in practice. 
We use 
%the state of the art SAT solver 
%cadical~\cite{Biere2019} via the interface pysat~\cite{pysat}
the SAT solver
picosat~\cite{Biere2019} via the Python interface pycosat~\cite{pysat}
%\footnote{Our framework can also use the art SAT solver cadical~\cite{Biere2019} via the interface pysat~\cite{pysat}, but the current version of pysat suffers from a memory leak which is fatal as the number of second level instances can become huge. We contacted the author to report this bug. }
to enumerate partial configurations,
which are then tested for modelling a gadget.
If a partial configuration does not
fulfill the properties of the desired gadget,
we prune the search space in a CDCL-like fashion 
to speed up the computations
(for more information on the conflict-driven clause learning algorithm, see \cite[Chapter~4]{HandbookSatisfiablity2009}).
Since the search space of all partial configurations is a partial order sorted by inclusion with respect to the domain, 
we exclude up- and down-sets which cannot contain a solution to the gadget-finding~problem.

\subparagraph{Results}
We use our algorithm to find gadgets for pattern avoiding sign mappings describing various combinatorial structures. 
Using our framework, we derive \NP-hardness proofs for thousands of structures. 
For rank $3$ we show that for at least \NUMPROOFS out of \NUMSETTINGS non-isomorphic settings for which our combination lemma applies the completion problem is \NP-hard (\Cref{thm:rank3}).  
For \NUMGREEDY of them, 
the completion problem can be decided in polynomial time in a greedy way. 
Since in rank~4 there are several millions of 
non-isomorphic settings,
we restrict our attention to a benchmark selection of \NUMSETTINGSFOURSELECTION settings.
From this selection, we classify \NUMPROOFSFOUR 
as \NP-hard (\Cref{thm:rank4}).
It is worth noting that our algorithm found gadgets for various other settings.
However, having a
combination lemma is crucial to combine the gadgets to a hardness reduction. 

By combining human ingenuity with computer power
we further managed to find an infinite family of sign mappings for which the completion problem is \NP-hard. This family consists of sign mappings of even rank $r \ge 4$ avoiding the two alternating sign sequences (\Cref{thm:even_rank}) which are a higher dimensional version of the interior triple systems introduced by Knuth.
These structures are 
a combinatorial generalization of point sets in dimension $r-1$. 
More specifically, they generalize the well-known acyclic chirotopes of rank~$r$ 
%as alternating sign sequences correspond to a positive/negative circuits in an oriented matroid 
\cite[Definition~3.4.7]{BjoenerLVWSZ1993};
for the rank~3 case see also \cite[Axiom~4]{Knuth1992}.

Before we describe our algorithm (\Cref{algo}) for solving the gadget-finding problem in Section~\ref{sec:algorithm}, 
we introduce some terminology
in Section~\ref{sec:prelim}
which allows to discuss completion problems in a unified manner.
In Section~\ref{ssec:sign_mappings},
we discuss completion problems on various pattern avoiding sign mappings and in \Cref{sec:thms} we present the settings where the application of our algorithm led to a hardness proof. 
In Section~\ref{sec:proof} we give a proof for the setting of generalized signotopes (a.k.a interior triple systems), a combinatorial structure arising from planar point sets \cite{Knuth1992}, pairwise intersecting
planar convex sets \cite{AgostonDKP22_convexsets} and simple drawings of the complete graph \cite{BFSSS_TDCTCG_2022}. 
The proofs for the remaining structures from \Cref{thm:rank3}, \Cref{thm:rank4}, and \Cref{thm:even_rank} will be discussed in \Cref{sec:others,app:even_rank}.
Essentially, all proofs work in a similar manner as the combination lemma (Lemma~\ref{lem:combination}) applies.
All gadgets are provided as supplemental data 
and can be verified with the provided framework \cite{supplemental_data}.

\subparagraph{Related work}

Automating the verification process of gadgets for NP-hardness reductions 
is not a new concept in theoretical computer science.
In certain cases it was even possible to construct gadgets in a (semi)automated manner.
Trevisan et al.\ \cite{TSSW2000}  used a linear-programming based framework to verify constant sized gadgets towards NP-hardness proofs for approximations of the MAXCUT problem.
%More specifically, to find suitable gadgets they start with a minimal configuration and set entries in a backtracking fashion until the configuration fulfills the desired properties or becomes "too strict" (i.e., setting further entries does not lead to a gadget).
%In this article we proceed with a similar idea but use a SAT solver for the enumeration and prune "too loose" configurations (i.e., unsetting entries does not lead to a gadget) from the search space in a more efficient manner.
%
%Renz and Li \cite{RenzLi2008} present an efficient procedure for automatically generating NP-hardness proofs for certain problems. 
%
Adler et al.\ \cite{ABDDLL2021} used a SAT/SMT-based framework to verify constant sized gadgets towards a NP-hardness proof of Tatamibari -- a japanese logic puzzle. 
%Since it was not specified how the gadgets were actually found, it is likely that a smart combination of pen-and-paper and computer assistance was~used.

% JUST BACHEOR+MASTER theses, no formal publications: Similarly \cite{Pulles2021} and \cite{Nijjar2004} developed framework to verify gadgets of reduction.

%\cite{Nijjar2004}  explored the possibility of automating NP-hardness reductions.  motivated the problem from an artificial intelligence perspective, then proposed the use of second-order existential (SO$\exists$) logic as representation language for decision problems. Building upon the theoretical framework of J. Antonio Medina (add cition also to that one?)
%https://scholarworks.umass.edu/dissertations/AAI9721479/

\section{Preliminaries}
\label{sec:prelim}

We consider the completion problem for \emph{sign mappings}, which play a central role in combinatorics and computational geometry~\cite{FelsnerGoodman2016, FelsnerWeil2001,GoodmanPollack1983,Knuth1992,ORourke1994_book}.
In particular, they are a key tool for computer-assisted investigations on geometric point sets; see e.g.\ \cite{AichholzerAurenhammerKrasser2001,SzekeresPeters2006}.
\mbox{Examples} are given in Section~\ref{ssec:sign_mappings}.

We use the standard notation $[n] = \{1,\ldots,n\}$ and $\binom{E}{r} = \{(x_1,\ldots,x_r) \subseteq E : x_1<\ldots<x_r\}$ for subsets $E \subseteq [n]$. Note that we only consider subsets $E$ of $[n]$ and hence there is a natural linear order on the elements.
A \emph{sign mapping on~$E$ of rank~$r$} is a mapping 
$\sigma:\binom{E}{r} \to \{+,-\}$ with $E \subseteq [n]$.
Such a mapping is \emph{partial}, if only some values of~$\binom{E}{r}$ are determined, i.e., a mapping $\sigma:\calD' \to \{+,-\}$ with $\calD' \subseteq \binom{E}{r}$.
In the following, we often consider a sign mapping $\sigma:\binom{E}{r} \to \{+,-\}$ as a word from $\{+,-\}^{\binom{|E|}{r}}$ with the convention that
signs are encoded in lexicographical order of the $r$-tuples. 
More generally, a partial mapping $\sigma:\calD' \to \{+,-\}$ with $\calD' \subseteq \binom{E}{r}$
can be considered as a mapping from $\binom{E}{r}$ to $\{+,?,-\}$ or as a word from $\{+,?,-\}^{\binom{|E|}{r}}$, where $?$ indicates that an entry is unset.
For example, the 4-element rank~3 partial sign mapping 
$\sigma(1,2,3)=+, \sigma(1,2,4)=-, \sigma(2,3,4)=+$ is encoded by ${+}{-}?{+}$.
A sign mapping $\sigma$ on~$[n]$ \emph{avoids} 
a family $\calF$ of sign patterns 
or is \emph{$\calF$-avoiding} for short 
if the word $\sigma|_{\binom{E'}{r}}$ is not contained in $\calF$ for all subsets $E' \subseteq E$. 
Note that we can use other domains such as $\calD = E^r$ in an analogous manner, which may be used to describe directed graphs, Latin squares, and Sudokus.
%Hence in the following we write $\calD$ for the domain of the considered sign mapping. 

\begin{problem}[$\calF$-\Complete]
    \label{problem:completability}
    Let $\sigma$ be a partial sign mapping on $ \calD' \subset \calD$.
    Is $\sigma$ \emph{completable} to an $\calF$-avoiding sign mapping on $\calD$, i.e., 
    is there an $\calF$-avoiding sign mapping  $\sigma^\ast$ on $\calD$ with $\sigma^\ast |_{\calD'} = \sigma$? 
\end{problem}

%For example, if the 4-element rank~3 partial mapping $\sigma$ encoded through ${+}{-}?{+}$ should be completed to a $\{{+}{-}{-}{+}\}$-avoiding mapping, the only possibility is to assign $\sigma(1,3,4) = {+}$.

In this article we prove \NP-hardness for several $\calF$-avoiding sign mappings via a reduction from 3SAT. 
For this we introduce gadgets, which are partial sign mappings on a small substructure with special completability properties.  
They are completable to a full $\calF$-avoiding sign mapping on the substructure if and only if their corresponding variables are a model for a Boolean formula.

%The goal of this article is to provide hardness proofs for specific completion problems via reductions from 3SAT.
%Before we can rigorously talk about "gadgets" for a reduction, 
%we will need to introduce some terminology.

More precisely, for a given sign mapping which is encoded on a domain $\calD$ with a family of forbidden substructures $\calF$, we want to find a gadget
that represents a quantifier-free Boolean formula $\psi$ on variables~$\calX$.
For example, we may search for a
propagator gadget $A \rightarrow B$,  a negator gadget $A \rightarrow \neg B$ and a clause gadget $A \vee B \vee C$. 

Since we connect the Boolean formula with sign mappings, 
we associate $+$ with $\True$ and $-$ with $\False$ 
for readability. 
To connect the variables of the Boolean formula $\psi$ 
with the given domain~$\calD$, 
we consider an injective function
$I:\calX \hookrightarrow \calD$.
The sign of the tuple~$I(x) \in \calD$ encodes the value of the variable~$x \in \calX$.
A \emph{gadget} is a partial mapping $\sigma : \calD' \to \{+,-\}$  
such that the signs of the $r$-tuples which correspond to the variables remain unset.
Moreover, for an assignment $f:\calX \to \{+,-\}$ of the variables of $\psi$, $\sigma_f$ denotes the mapping obtained from $\sigma$ 
by further setting the values of the substructures corresponding to the assignment of the variables in the Boolean formula:
$\sigma_f(I(x)) := f(x)$ for all $x \in \calX$.
\iffalse
\begin{align*}
    \sigma_f: \calD' \cup I(\calX) \to \{+,-\} \\
    \sigma_f(I(x)) := f(x)
    \sigma_f(d) = \sigma(d) \text{ for all } d\in \calD'
\end{align*}
\fi
A partial mapping $\sigma$ is a \emph{gadget} if 
for every possible assignment of the variables 
$f:\calX \to \{+,-\}$, 
$\sigma_f$ is completable if and only if $\psi(f)=\True$.

\subsection{Examples of sign mappings}
\label{ssec:sign_mappings}

Many well-studied combinatorial designs such as matroids or Steiner systems can be described
via pattern-avoiding sign mappings on some domain $\binom{E}{r}$ which 
indicate which $r$-tuples form a basis or block, respectively.
More specifically,
Steiner systems $S(t,k,n)$ 
can be described in terms of forbidden subconfigurations on $(2k-t)$-tuples,
and the basis-exchange property of rank~$r$ 
matroids yield forbidden subconfigurations on $2r$-element subsets.
While the completion problem was proven to be hard for several well-known combinatorial designs such as Steiner systems, Latin squares, or Sudokus \cite{Colbourn1983,Colbourn1984,YatoSeta2003},
resolving the complexity appears to a non-trivial task in general and the complexity remains unknown for many well-known structures such as matroids.

Before we give a quantitative characterization, 
we discuss some 
simple examples of sign mappings on the 
domain $\calD = \binom{[n]}{r}$ which play an central role in combinatorics and computational geometry.

\subparagraph{Permutations}

A permutation $\pi : [n] \to [n]$ 
is uniquely determined by its set of inversions
$\Inv(\pi) := \{(a,b) \in \binom{[n]}{2} : \pi(a) > \pi(b)\}$,
which have the property that if $(a,b),(b,c) \in \Inv(\pi)$, then $(a,c) \in \Inv(\pi)$ and similarly if $(a,b),(b,c) \not \in \Inv(\pi)$, then $(a,c) \not \in \Inv(\pi)$. 
By assigning $+$ to each pair $(a,b)$ with $(a,b) \notin \Inv(\pi)$ and $-$ otherwise, we obtain a rank~2 sign mapping $\sigma:\binom{[n]}{2} \to \{+,-\}$ such that for no three distinct elements $a,b,c \in \binom{[n]}{3}$ it holds $\sigma(a,b)={+}$, $\sigma(a,c)={-}$, $\sigma(b,c)={+}$ or $\sigma(a,b)={-}$, $\sigma(a,c)={+}$, $\sigma(b,c)={-}$.
In other words, $\sigma$ is $\{{+}{-}{+},{-}{+}{-}\}$-avoiding. This is a well-known combinatorial encoding of permutations. 

Since we can check efficiently whether a partial mapping avoids the patterns ${+}{-}{+}$ and ${-}{+}{-}$,
and such partial mappings are in correspondence with partial orders on~$[n]$ we can always complete them to a total order on $[n]$ which correspond to a permutation. 
Hence $\{{+}{-}{+},{-}{+}{-}\}$-\Complete can be solved efficiently.

\subparagraph{Cyclic order}
In a \emph{cyclic order},
the elements of $[n]$ are arranged on a circle. 
This order corresponds to a cyclic permutation and is determined by its \emph{cyclic inversions}, i.e., the triples $(a,b,c)$ with $a<b<c$ which appear in the order $a,c,b$ in the cycle.
%Similar as in the non-cyclic setting, 
%a cyclic permutation $\pi: [n] \to [n]$ is determined by the set of its cyclic inversions,
%i.e., triples $a,b,c \in \binom{[n]}{3}$ for which $\pi(c),\pi(b),\pi(a)$ appear in this particular cyclic order.
Cyclic orders on $[n]$ are in correspondence with 
rank~3 sign mappings on $[n]$ avoiding 
$\{{+}{-}{+}{-},\,
{-}{+}{-}{+},\, \allowbreak{}
{+}{-}{-}{-},\, 
{-}{+}{-}{-},\, \allowbreak{}
{-}{-}{+}{-},\, 
{-}{-}{-}{+},\, \allowbreak{}
{-}{+}{+}{+},\,
{+}{-}{+}{+},\,\allowbreak{}
{+}{+}{-}{+},\,$  ${+}{+}{+}{-}\}$.
In contrast to the setting of permutation
where one can efficiently decide completability, 
Galil and Megiddo~\cite{GalilMegiddo1977} showed that \Complete is \NP-complete for cyclic orders.

\subparagraph{CC-Systems}
Another well-known example of rank~3 sign-mappings appears as a natural relaxation of point sets in the plane.
In literature, this structure occurs under various names such as \emph{CC-systems}, \emph{abstract order types},  or \emph{acyclic rank~3 chirotopes}.
They can be characterized
by forbidden sign patterns on 5-element subsets.
We refer the interested reader to Knuth's book \cite{Knuth1992},
where he shows that \Complete is \NP-complete for CC-systems.
Knuth's proof also transfers to the setting of pre-CC-systems
\cite{Baier05_chirotopesNP},
which are known as rank~3 chirotopes, not necessary acyclic.
For an independent proof see also \cite[Chapter~5]{Tschirschnitz2003}.

\subparagraph{Generalized Signotopes}

Towards an axiomatization for point sets in the plane,
Knuth~\cite{Knuth1992} observed that the triple-orientations of a point set $S$ in the plane induce a rank~3 sign mapping $\sigma_S$ which in particular avoids the two patterns ${+}{-}{+}{-}$ and ${-}{+}{-}{+}$.
In the want of a better name,
Knuth named such rank~3 sign mappings 
%avoiding $\{{+}{-}{+}{-},{-}{+}{-}{+}\}$
\emph{interior triple systems}.
Following \cite{BFSSS_TDCTCG_2022}, we shall name them \emph{generalized signotopes},
as they are a natural generalization of  signotopes (of rank~3). 
Generalized signotopes not only
generalize cyclic permutations and CC-\-systems, but appear in the context of various combinatorial structures such as 
simple topological drawings~\cite{BFSSS_TDCTCG_2022}, monotone colorings of hypergraphs~\cite{Balko2019} and convex set systems \cite{AgostonDKP22_convexsets}.

More generally, it is well-known that the sign mapping $\sigma_S$ of a point set $S$ in dimension $d$ is of rank $r=d+1$ and fulfills the axioms of an acyclic  chirotope \cite[Definition~3.4.7]{BjoenerLVWSZ1993}.
In particular, it 
avoids the two alternating patterns ${+}{-}{+}{-} \ldots ({-})^{d+1}$ and ${-}{+}{-}{+} \ldots ({-})^{d}$ of length $d+2$, %, which correspond to a positive/negative circuits in the oriented matroid. \helena{das würde ich weglassen, aber keine ahnung}
where $(-)^k$ denotes the sign of $(-1)^k$. 
This is a higher dimensional version of generalized signotopes. 
Recall that rank~$2$ sign mappings avoiding $\{{+}{-}{+}, {-}{+}{-}\}$ are permutations.
%\helena{not sure whether we want to keep that much about signotopes}
%A \emph{signotope of rank~$3$} is a sign mapping avoiding $\calF = \{
%{+}{-}{+}{-},\allowbreak{} 
%{-}{+}{-}{+},\allowbreak{} 
%{+}{-}{+}{+},\allowbreak{} 
%{-}{+}{-}{-},\allowbreak{} 
%{+}{-}{-}{+},\allowbreak{} 
%{-}{+}{+}{-},\allowbreak{} 
%{+}{+}{-}{+},\allowbreak{} {-}{-}{+}{-}\}$, 
%i.e., in the sign sequence corresponding to the 4-tuples there is at most one sign change. 
%More generally, a sign mapping of rank~$r$ is called \emph{signotope} if all pattern sequences of length $r+1$ with more than one sign change are avoided. 
%Note that signotopes of rank~2 are permutations 
%and that rank~3 signotopes are related to CC-systems, 
%pseudoline arrangements in the plane~\cite{FelsnerWeil2001},
%and pseudolinear drawings of the complete graph~\cite{BalkoFulekKyncl2015}.

\section{Hard structures}
\label{sec:thms}

For rank $r=3$, we show that $\calF$-\Complete is \NP-complete for \NUMPROOFS non-isomorphic settings~$\calF$ such that no sign sequence with two consecutive plus signs is contained in $\calF$.

\begin{theorem}
    \label{thm:rank3}
    For $r=3$, $\calF$-\Complete is \NP-complete 
    for \NUMPROOFS families $\calF$  
    which are given in \Cref{listing:families}.
    This includes generalized signotopes, i.e., $\calF = \{{+}{-}{+}{-},{-}{+}{-}{+}\}$.
\end{theorem}

\begin{lstlisting}[
backgroundcolor=\color{white},
    breaklines=true,
    breakindent=0pt,
    basicstyle=\ttfamily,
    %xleftmargin=.25in,
    %xrightmargin=.25in,
    caption={\NUMPROOFS families for $r =3$ for which \Complete is \NP-complete.},
    mathescape,
    label=listing:families,
    ]
$\{{+}{-}{+}{-},{+}{-}{-}{+}\}$, $\{{+}{-}{+}{-},{-}{+}{-}{+}\}$, $\{{+}{-}{+}{-},{-}{+}{-}{-}\}$, $\{{+}{-}{+}{-},{-}{-}{-}{+}\}$, $\{{+}{-}{+}{-},{-}{-}{-}{-}\}$, $\{{+}{-}{-}{+},{-}{+}{-}{-}\}$, $\{{+}{-}{-}{+},{-}{-}{-}{-}\}$, $\{{+}{-}{+}{-},{+}{-}{-}{+},{-}{+}{-}{+}\}$, $\{{+}{-}{+}{-},{+}{-}{-}{+},{-}{+}{-}{-}\}$, $\{{+}{-}{+}{-},{+}{-}{-}{+},{-}{-}{+}{-}\}$, $\{{+}{-}{+}{-},{+}{-}{-}{+},{-}{-}{-}{+}\}$, $\{{+}{-}{+}{-},{+}{-}{-}{+},{-}{-}{-}{-}\}$, $\{{+}{-}{+}{-},{+}{-}{-}{-},{-}{+}{-}{+}\}$, $\{{+}{-}{+}{-},{-}{+}{-}{+},{-}{+}{-}{-}\}$, $\{{+}{-}{+}{-},{-}{+}{-}{+},{-}{-}{-}{-}\}$, $\{{+}{-}{+}{-},{-}{+}{-}{-},{-}{-}{-}{+}\}$, $\{{+}{-}{+}{-},{-}{+}{-}{-},{-}{-}{-}{-}\}$, $\{{+}{-}{+}{-},{-}{-}{-}{+},{-}{-}{-}{-}\}$, $\{{+}{-}{-}{+},{-}{+}{-}{-},{-}{-}{+}{-}\}$, $\{{+}{-}{-}{+},{-}{+}{-}{-},{-}{-}{-}{-}\}$, $\{{+}{-}{+}{-},{+}{-}{-}{+},{+}{-}{-}{-},{-}{+}{-}{+}\}$, $\{{+}{-}{+}{-},{+}{-}{-}{+},{-}{+}{-}{+},{-}{+}{-}{-}\}$, $\{{+}{-}{+}{-},{+}{-}{-}{+},{-}{+}{-}{+},{-}{-}{-}{-}\}$, $\{{+}{-}{+}{-},{+}{-}{-}{+},{-}{+}{-}{-},{-}{-}{+}{-}\}$, $\{{+}{-}{+}{-},{+}{-}{-}{+},{-}{+}{-}{-},{-}{-}{-}{+}\}$, $\{{+}{-}{+}{-},{+}{-}{-}{+},{-}{+}{-}{-},{-}{-}{-}{-}\}$, $\{{+}{-}{+}{-},{+}{-}{-}{+},{-}{-}{+}{-},{-}{-}{-}{-}\}$, $\{{+}{-}{+}{-},{+}{-}{-}{+},{-}{-}{-}{+},{-}{-}{-}{-}\}$, $\{{+}{-}{+}{-},{+}{-}{-}{-},{-}{+}{-}{+},{-}{-}{+}{-}\}$, $\{{+}{-}{+}{-},{+}{-}{-}{-},{-}{+}{-}{+},{-}{-}{-}{-}\}$, $\{{+}{-}{+}{-},{-}{+}{-}{+},{-}{+}{-}{-},{-}{-}{-}{-}\}$, $\{{+}{-}{+}{-},{-}{+}{-}{-},{-}{-}{-}{+},{-}{-}{-}{-}\}$, $\{{+}{-}{-}{+},{-}{+}{-}{-},{-}{-}{+}{-},{-}{-}{-}{-}\}$, $\{{+}{-}{+}{-},{+}{-}{-}{+},{+}{-}{-}{-},{-}{+}{-}{+},{-}{-}{+}{-}\}$, 
$\{{+}{-}{+}{-},{+}{-}{-}{+},{+}{-}{-}{-},{-}{+}{-}{+},{-}{-}{-}{-}\}$, $\{{+}{-}{+}{-},{+}{-}{-}{+},{-}{+}{-}{+},{-}{+}{-}{-},{-}{-}{+}{-}\}$, 
$\{{+}{-}{+}{-},{+}{-}{-}{+},{-}{+}{-}{+},{-}{+}{-}{-},{-}{-}{-}{-}\}$, $\{{+}{-}{+}{-},{+}{-}{-}{+},{-}{+}{-}{-},{-}{-}{+}{-},{-}{-}{-}{-}\}$, 
$\{{+}{-}{+}{-},{+}{-}{-}{+},{-}{+}{-}{-},{-}{-}{-}{+},{-}{-}{-}{-}\}$, $\{{+}{-}{+}{-},{+}{-}{-}{+},{+}{-}{-}{-},{-}{+}{-}{+},{-}{-}{+}{-},{-}{-}{-}{-}\}$, $\{{+}{-}{+}{-},{+}{-}{-}{+},{-}{+}{-}{+},{-}{+}{-}{-},{-}{-}{+}{-},{-}{-}{-}{-}\}$
\end{lstlisting}

The proof of \Cref{thm:rank3} is divided into two parts. In \Cref{sec:proof} we give the full proof in the setting of generalized signotopes. 
To show \NP-hardness, we reduce from 3SAT.
For this we define propagator and clause gadgets as partial mappings on small domains. 
The combination of the gadgets is possible by Lemma~\ref{lem:combination}.
The combination lemma only applies to \NUMSETTINGS non-isomorphic structures of rank 3.
For all \NUMSETTINGS we checked whether there exists gadgets for our reduction using a simple algorithm based on a SAT framework to find gadgets. The algorithm is described in \Cref{sec:algorithm} and in \Cref{sec:others} we explain how to choose the gadgets that they work for a reduction using our combination lemma.
While for \NUMPROOFS structures \Complete turns out to be \NP-hard, \Complete can be solved efficiently for \NUMGREEDY structures.
More specifically, Listing~\ref{listing:greedy} lists \NUMGREEDY of the \NUMSETTINGS settings where \Complete can be solved 
by assigning~$+$ to unset tuples assuming that there are no sequences in which the remaining sign are determined by the forbidden patterns.

\iffalse
To combine the gadgets, which were found by computer using Algorithm~\ref{algo},
we need a combination lemma (Lemma~\ref{lem:combination}).
In order to use this combination lemma, we restrict the gadget finding problems to those gadgets where the variables are encoded in a triple of three consecutive elements. Moreover, we only tested gadgets of up to size~6. 
Our combination lemma (see Lemma~\ref{lem:combination})  applies to \NUMSETTINGS non-isomorphic structures.
While for \NUMPROOFS structures \Complete turns out to be hard, \Complete can be solved efficiently for \NUMGREEDY structures.
More specifically, Listing~\ref{listing:greedy} lists \NUMGREEDY of the \NUMSETTINGS settings where \Complete can be solved 
by assigning~$+$ to unset tuples assuming that there are no sequences in which the remaining sign are determined by the forbidden patterns.
\fi

\goodbreak

\begin{lstlisting}[
backgroundcolor=\color{white},
    breaklines=true,
    breakindent=0pt,
    basicstyle=\ttfamily,
    %xleftmargin=.25in,
    %xrightmargin=.25in,
    caption={\NUMGREEDY families for $r=3$ 
        for which \Complete can be solved by greedily assigning~$+$.},
    mathescape,
    label=listing:greedy,
    ]
$\emptyset$, $\{{-}{-}{-}{-}\}$, $\{{+}{-}{-}{-}\}$, $\{{-}{+}{-}{-}\}$, $\{{+}{-}{-}{-},{-}{+}{-}{-}\}$, $\{{+}{-}{-}{-}, {-}{-}{+}{-}\}$, $\{{+}{-}{-}{-}, {-}{-}{-}{+}\}$, $\{{+}{-}{-}{-}, {-}{-}{-}{-}\}$, $\{{-}{+}{-}{-},{-}{-}{+}{-}\}$, $\{{-}{+}{-}{-}, {-}{-}{-}{-}\}$, $\{{+}{-}{-}{-}, {-}{+}{-}{-}, {-}{-}{+}{-}\}$, $\{{+}{-}{-}{-}, {-}{+}{-}{-}, {-}{-}{-}{+}\}$, $\{{+}{-}{-}{-}, {-}{+}{-}{-}, {-}{-}{-}{-}\}$, $\{{+}{-}{-}{-}, {-}{-}{+}{-}, {-}{-}{-}{-}\}$, $\{{+}{-}{-}{-}, {-}{-}{-}{+}, {-}{-}{-}{-}\}$, $\{{-}{+}{-}{-}, {-}{-}{+}{-}, {-}{-}{-}{-}\}$, $\{{+}{-}{-}{-}, {-}{+}{-}{-}, {-}{-}{+}{-}, {-}{-}{-}{+}\}$, $\{{+}{-}{-}{-}, {-}{+}{-}{-}, {-}{-}{+}{-}, {-}{-}{-}{-}\}$, $\{{+}{-}{-}{-}, {-}{+}{-}{-}, {-}{-}{-}{+}, {-}{-}{-}{-}\}$, $\{{+}{-}{-}{-}, {-}{+}{-}{-}, {-}{-}{+}{-}, {-}{-}{-}{+}, {-}{-}{-}{-}\}$
\end{lstlisting}

%Unfortunately, the complexity of \Complete for signotopes, which was posed as a question by Felsner, G{\"{a}}rtner and Tschirschnitz \cite{FelsnerGT05}, remains open. 
%\manfred{i removed the question by FelsnerGT because this looks like a weakness}

For rank $r=4$ there are millions of settings where all sign patterns of length 5 with 3 consecutive pluses are allowed. To showcase the power of our framework, we restricted our attention to the subclass in which all sign patterns with 2 consecutive pluses signs are allowed. Among the  \NUMSETTINGSFOURSELECTION non-isomorphic settings, our framework classified \NUMPROOFSFOUR as \NP-hard. 
All gadgets are given as supplemental data~\cite{supplemental_data}.

\begin{theorem}
\label{thm:rank4}
    For $r = 4$, $\calF$-\Complete is \NP-complete for \NUMPROOFSFOUR families $\calF$. % which are provided as supplemental data~\cite{supplemental_data}.
\end{theorem}

From the gadgets in ranks 3 and~4, we managed to derive the a construction for general even rank $r$ where we avoid the two alternating sign pattern of length $r+1$.
The explicit description of the gadgets for the reduction is deferred to~\Cref{app:even_rank}.

\begin{theorem}
\label{thm:even_rank}
    For even $r \geq 4$, $\calF$-\Complete is \NP-complete with %$\calF = \{{+}({-}{+})^{r/2}, {-}({+}{-})^{r/2} \}$.
    $
    \calF = \{
    \underbrace{
    {+}{-}{+}\ldots{+}
    }_{\text{$r+1$ signs}}
    ,
    \underbrace{
    {-}{+}{-}\ldots{-}
    }_{\text{$r+1$ signs}}
    \}.
    $
\end{theorem}

\section{Proof of Theorem~\ref{thm:rank3} for generalized signotopes}
\label{sec:proof}
In this section, we prove Theorem~\ref{thm:rank3}.   
The completion problem is clearly contained in~\NP,
so it remains to prove the hardness,
for which we perform a reduction from 3SAT.
In the following we assume without loss of generality that no clause contains the same variable more than once as
any 3SAT formula can be extended by a linear number of auxiliary variables and clauses to ensure this property.    
The main idea is to introduce gadgets which are then combined using a \emph{combination lemma} (Lemma~\ref{lem:combination}).
Since the lemma works for all $\calF$-avoiding  rank 3 mappings with 
$\calF \subseteq \{
{+}{-}{+}{-},\allowbreak{}
{+}{-}{-}{+},\allowbreak{}
{+}{-}{-}{-},\allowbreak{}
{-}{+}{-}{+},\allowbreak{}
{-}{+}{-}{-},\allowbreak{}
{-}{-}{+}{-},\allowbreak{}
{-}{-}{-}{+},\allowbreak{}
{-}{-}{-}{-}\}$,
it only remains to find gadgets for all \NUMPROOFS setting to conclude \NP-completeness. 
To do this in a fully automated way, 
we used the algorithm described in Section~\ref{sec:algorithm}.
When the gadgets are known, they can be verified quite easily.
All settings for which we found the gadgets are given in \Cref{listing:families}.
The gadgets, a verification program, and a program to verify that all settings are non-isomorphic are provided as supplemental data.%~\cite{supplemental_data}.

In the following we give a formal proof
for generalized signotopes, that is, $\{{+}{-}{+}{-},\allowbreak {-}{+}{-}{+}\}$-avoiding sign mappings.
All other settings will be analogous; see Section~\ref{sec:others}.

Let $\phi$ be an instance of 3SAT in conjunctive normal form (CNF) on $n$ variables $v_1,\ldots,v_n$ with $m$ clauses 
$c_i = \ell_{i,1} \vee \ell_{i,2} \vee \ell_{i,3}$ with literals $\ell_{i,j} = v_k$ or $\ell_{i,j} = \neg v_k$ for some $k=k(i,j)$.
We create a partial generalized signotope $\sigma_\phi$ on $N=3n+5m$ elements
such that for each model $M$ of $\phi$ (i.e.\ an assignment to the variables such that all clauses are fulfilled)
there exists a completion $\sigma$ of $\sigma_\phi$ and vice versa.
We encode the values of a variable $v_i$ in the triple $V_i = (3i-2,3i-1,3i)$. % with  $v_i:=3i-2$,
that is,
$M(v_i)=\True$ if and only if 
$\sigma(V_i)=+$.    
A~clause $c_i$ is encoded by a gadget using triples of the 5-element subset
$C_i = \{3n+5i-4,\ldots,3n+5i\}$. % with $c_i := 3n+5i-4$.  
The literals $\ell_{i,j}$ are stored in a triple in the clause gadget and connected to their variable counterpart in $[3n]$ such that 
the value from the clause gadget is synchronized with the variable gadget via propagator gadgets in the sense that, if $\ell_{i,j}$ is assigned to \True in~$M$,
then the variable $v_{k}$ must be set \True if $\ell_{i,j} = v_{k}$ and \False if $\ell_{i,j} = \neg v_{k}$.
The synchronization between the Boolean value and the sign of a triple in the partial mapping is such that \True is~$+$ and \False is~$-$.

\subsection{Clause-Gadget \texorpdfstring{$CG(X_1{+} \vee X_2{-} \vee X_3{+})$}{CG(X1+ v X2+ v X3+)}}

To encode a clause $x_1 \vee \neg x_2 \vee x_3$ of $\phi$ on the 5-element subset $C = \{c,\ldots, c+4\}$,
we encode the value of the literal $x_i$ in the triple $X_i := (c+i-1,c+i,c+i+1)$ for $i =1,2,3$.
The clause gadget is constructed such that the partial generalized signotope $\sigma_\phi$ restricted to the elements $\{c, \ldots, c+4\}$, denoted by $\sigma_C$,  is completable to $\sigma$ if and only if $\sigma(X_1) = + \vee \sigma(X_2)=- \vee \sigma(X_3) = +$. 
Note that it is irrelevant that the second literal of the clause is negated because we can negate it again later as part of a propagator gadget.
For the clause gadget we assign the following signs
\begin{align*}
    &\sigma_C(c,c+1,c+3)=+, \                 &&\sigma_C(c,c+1,c+4)=-, \\
    &\sigma_C(c,c+2,c+4)=-, \
    &&\sigma_C(c, c+3,c+4)=+,\\
    &\sigma_C(c+1,c+2,c+4) =+,\
    &&\sigma_C(c+1,c+3,c+4)=+,
\end{align*}
which have the following property:

\begin{claim}
    Every completion $\sigma$ of $\sigma_C$ fulfills $\sigma(X_1) = {+} \vee \sigma(X_2) = - \vee \sigma(X_3) = {+}$. 
    Moreover, 
    if we assign $\sigma_C(X_1) = {+}$, $\sigma_C(X_2) = -$, or $\sigma_C(X_3) = {+}$, then
    there exists a completion $\sigma$.
\end{claim}

\begin{proof}
    In a first step we consider the set $\{c,c+1,c+2,c+3\}$ which has the sequence
    $?{+}??$ since only the triple $\{c,c+1,c+3\}$ is assigned. 
    The only possible way to assign values such that this 4-subset becomes invalid, is to complete it to ${-}{+}{-}{+}$, which means
    $\sigma(X_1) = \sigma(c,c+1,c+2) = -$ and $\sigma(c,c+2,c+3) = -$ and $\sigma(X_2) = \sigma(c+1,c+2,c+3) = +$.
    Hence in a completion $\sigma$ it holds
    $\sigma(X_1) = \sigma(c,c+1,c+2) = {+}$ or $\sigma(c,c+2,c+3) = {+}$ or $\sigma(X_2) = \sigma(c+1,c+2,c+3) = -$.
    In a next step, we look at the set $\{c,c+2,c+3,c+4\}$
    to make a connection between the triples $(c,c+2,c+3)$ and $X_3 = (c+2,c+3,c+4)$.
    The sign sequence of this set is $?-{+}?$.
    Again the there is only one possibility to make this an avoiding sequence. Hence in every completion $\sigma$ it is
    $\sigma(c,c+2,c+3) = -$ or $\sigma(c+2,c+3,c+4) = {+}$. 
    In other words $\sigma(c,c+2,c+3)={+}$ implies $\sigma(X_3) = \sigma(c+2,c+3,c+4)={+}$.
    Together this shows that for every completion $\sigma$ it holds $\sigma(X_1) = {+} \allowbreak{} \vee \sigma(X_2) = - \vee \sigma(X_3) = {+}$.
    On the other hand, note that the other $4$-tuples are always fine since 
    \begin{align*}
        \sigma_C|_{(c,c+1,c+2,c+4)} = ?{-}{-}{+}, \ 
        \sigma_C|_{(c,c+1,c+3,c+4)}={+}{-}{+}{+}, \ % \text{ and }
        \sigma_C|_{(c+1,c+2,c+3,c+4)}= ?{+}{+}?. 
    \end{align*}
    This shows that if the clause is \True, then there is always a completion. 
\end{proof}

\subsection{Propagator-Gadget \texorpdfstring{$PG(a_2a_3a_4{+} \rightarrow a_1a_2a_3{+})$}{PG(a2a3a4+ -> a1a2a3+)}} 

For this propagation, 
we assign values to triples of the set 
$P=\{a_1, \ldots, a_4\}$. 
The restricted mapping is denoted by $\sigma_P$. We assign
$ \sigma_P(a_1,a_2,a_4) ={+}$ and $\sigma_P(a_1,a_3,a_4) = {-}$.

\begin{claim}
    Every completion $\sigma$ of $\sigma_P$ with 
    $\sigma(a_2,a_3,a_4) = {+}$ fulfills $\sigma(a_1,a_2,a_3) = {+}$.
    Moreover, 
    if we additionally assign $\sigma_P(a_2,a_3,a_4) = {-}$ or $\sigma_P(a_1,a_2,a_3) = {+}$, then
    there exists a completion~$\sigma$.
\end{claim}

\begin{proof}
    The 4-subset $a_1, a_2,a_3,a_4$ has the sign sequence $?{+}{-}?$. 
    Hence in order to avoid the forbidden pattern ${-}{+}{-}{+}$, assigning $\sigma(a_2,a_3,a_4) = {+}$ implies $\sigma(a_1,a_2,a_3) = {+}$.
\end{proof}

Additionally,  a propagator gadget for $\sigma(a_2,a_3,a_4) = {-} \rightarrow \sigma(a_1,a_2,a_3) = {-}$ is obtained by the opposite signs of the triple orientations:  $\sigma_P(a_1,a_2,a_4) = {+}$,  $\sigma_P(a_1,a_3,a_4) = {-}$ (short: $PG(a_2a_3a_4{-} \rightarrow a_1a_2a_3{-})$).

\subsection{Propagator-Gadget  \texorpdfstring{$PG(a_3a_4a_5{-} \rightarrow a_1a_2a_3 {+})$}{PG(a3a4a5- -> a1a2a3+)}}
In order to negate variables, we also need to negate the value of a triple.
We encode the gadget on the five elements $ P = \{a_1, \ldots, a_5\}$
and assign:
\begin{align*}
    &\sigma_P(a_1,a_2,a_4) = {+},\
    &\sigma_P(a_1,a_2,a_5) = {-},\
    &&\sigma_P(a_1,a_3,a_5) = {-},\\
    &\sigma_P(a_1,a_4,a_5) = {+},\
    &\sigma_P(a_2,a_3,a_4) = {+},\
    &&\sigma_P(a_2,a_3,a_5) = {+},\\
    &\sigma_P(a_2,a_4,a_5) = {+}.
\end{align*}

\begin{claim}
    Every completion $\sigma$ of $\sigma_P$ with 
    $\sigma(a_3,a_4,a_5) = {-}$ fulfills $\sigma(a_1,a_2,a_3) = {+}$.
    Moreover, 
    if we additionally assign $\sigma_P(a_3,a_4,a_5) = {+}$ or $\sigma_P(a_1,a_2,a_3) = {+}$, then
    there exists a completion~$\sigma$.
\end{claim}

\begin{proof}
    Combine the conditions of the 4-tuples $a_1,a_2,a_3,a_4$ and $a_1,a_3,a_4,a_5$.
    The sequence of the latter one is $?{-}{+}?$. Hence $\sigma(a_3,a_4,a_5) = {-}$ implies that $\sigma(a_1,a_3,a_4) = {-}$.
    Furthermore the sign sequence of $a_1,a_2,a_3,a_4$ is $?{+}?{+}$, which shows that $\sigma(a_1,a_3,a_4) = {-}$  implies $\sigma(a_1,a_2,a_3) = {+}$. 
    All other 4-tuples are always fine since
    \begin{align*}
        \sigma_P|_{(a_1,a_2,a_3,a_5)} = ?{-}{-}{+},  \quad 
        \sigma_P|_{(a_1,a_2,a_4,a_5)} =  {+}{-}{+}{+},  \quad 
        \sigma_P|_{(a_2,a_3,a_4,a_5)}= {+}{+}{+}? 
    \end{align*}
    which shows the moreover part.
\end{proof}

Again,  a propagator gadget for $\sigma(a_3,a_4,a_5) = {+} \rightarrow \sigma(a_1,a_2,a_3) = {-}$ is modeled by the opposite signs of the triple orientations and denoted by $PG(a_3a_4a_5{+} \rightarrow a_1a_2a_3{-})$.\\

For the sake of readability we introduce larger propagator gadgets consisting of two or three already defined propagator gadgets. 
In order to propagate a value from a clause to a variable, we need a propagator gadget which is able to propagate the value from the triple $(b_1,b_2,b_3)$ to the triple $(a_1,a_2,a_3)$ for every $a_1<a_2<a_3<b_1<b_2<b_3$.
For this we define the following four propagator gadgets consisting of the already introduced propagator gadgets as illustrated in Figure~\ref{fig:propagators}. 

\begin{figure}[htb]
\centering
    \begin{subfigure}[t]{.46\textwidth}
        \centering
        \includegraphics[page=1]{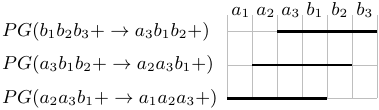}
        \caption{$PG(b_1b_2b_3 {+} \rightarrow a_1a_2a_3 {+})$}
        \label{subfig:prop++}
    \end{subfigure} 
% \\[2ex]
\qquad
    \begin{subfigure}[t]{.46\textwidth}
        \centering
        \includegraphics[page =2]{propagator_6_small.pdf}
        \caption{$PG(b_1b_2b_3 {-} \rightarrow a_1a_2a_3 {-})$}
        \label{subfig:prop--}
    \end{subfigure} 
 \\[2ex]
 
    \begin{subfigure}[t]{.46\textwidth}
        \centering
        \includegraphics[page =3]{propagator_6_small.pdf}
        \caption{$PG(b_1b_2b_3 {+} \rightarrow a_1a_2a_3 {-})$}
        \label{subfig:prop+-}
    \end{subfigure} 
% \\[2ex]
\qquad
    \begin{subfigure}[t]{.46\textwidth}
        \centering
        \includegraphics[page =4]{propagator_6_small.pdf}
        \caption{$PG(b_1b_2b_3 {-} \rightarrow a_1a_2a_3 {+})$}
        \label{subfig:prop-+}
    \end{subfigure} 
    \caption{Construction of the four propagator gadgets on 6 elements. }
    \label{fig:propagators}
\end{figure}

\subsection{Combination of Gadgets}

To build the desired partial generalized signotope $\sigma_\phi$, we add a clause gadget on the 5-element subset $C_i = \{3n+5i-4,\ldots,3n+5i\}$ for each clause $c_i$.
For each literal, we propagate either the same value or the opposite value to the corresponding variable. 
Hence, we add one of the four propagator gadgets depending on the literal $\ell_{i,j}$.
Recall that $\ell_{i,1}$ corresponds to $L_{i,1} = (c_i,c_i+1,c_i+2)$, $\neg \ell_{i_2}$ to $L_{i,2}= (c_i+1,c_i+2,c_i+3)$, and $\ell_{i,3}$ to $L_{i_3} =(c_i+2,c_i+3,c_i+4)$.
\begin{align*}
    \ell_{i,j}=v_k:\ &PG(L_{i,j}{+} \rightarrow V_k {+}) \text{ for } j\in\{1,3\},  \\
 &PG(L_{i,2} {-} \rightarrow V_k {+})\\ 
    \ell_{i,j}=\neg v_k:\   &PG(L_{i,j} {+} \rightarrow V_k {-}) \text{ for } j\in\{1,3\}, \\
 &PG(L_{i,2} {-} \rightarrow V_k{-})
\end{align*}
With this it holds
\begin{claim}
    If there is a completion $\sigma$ of $\sigma_\phi$, then there is a model $M$ of $\phi$.
\end{claim}

However, for the reverse we need  to complete the signs of triples of different gadgets. 
Let $A\subset B \subseteq [N]$. Then $A$ is \emph{consecutive} within $B$ if no $b\in B\setminus A$ and $a_1,a_2\in A$ exist, such that $a_1<b<a_2$.

\begin{lemma}[The Combination Lemma]
    \label{lem:combination}
    Let $r \ge 3$, let
 $\calF$ be a set of words of length $r+1$ on the symbols $\{+,-\}$
 such that no word contains $r-1$ consecutive $+$ symbols,
 and let $\mathcal{I}\subseteq 2^{[N]}$ such that $I\cap J$ is consecutive within $I\cup J$ for every $I,J \in \mathcal{I}$.
    If there exist
    $\calF$-avoiding mappings $\sigma_I$ on~$I$ for all $I \in \mathcal{I}$
    such that $\sigma_I$ and $\sigma_J$ agree on        $I\cap J$ for every $I,J \in \mathcal{I}$, 
    then
    the following rank~$r$ sign mapping on~$[N]$ is $\calF$-avoiding:
    \begin{align*}
        \sigma(X) = \begin{cases}
            \sigma_I(X) & \text{if } X \subseteq I \text{ for an } I  \in \mathcal{I}\\
            {+} & \text{otherwise.}
        \end{cases}
    \end{align*}
\end{lemma}

Note that for $r =3$ the lemma requires that the 4-subsets in which two consecutive ${+}$ signs appear in the sign sequence are valid. 
This is the case if 
$\calF \subseteq \{{+}{-}{+}{-},{+}{-}{-}{+},\allowbreak{}{+}{-}{-}{-},\allowbreak{}{-}{+}{-}{+},\allowbreak{}{-}{+}{-}{-},\allowbreak{}{-}{-}{+}{-},\allowbreak{}{-}{-}{-}{+},\allowbreak {-}{-}{-}{-}\}$.
Also note that by reversing all signs in the statement of Lemma~\ref{lem:combination} (i.e.,\  exchange the roles of $+$ and~$-$), one obtains an analogous combination lemma.

\begin{proof}
To show this, we check the $(r+1)$-subsets $P=\{x_1 <x_2,\ldots<x_{r+1}\}$. 
If no $I \in \mathcal{I}$ contains $r$ elements of $P$, then the sign sequence corresponding to $P$ is clearly $\calF$-avoiding.
If $P \subseteq I$ for some $I \in \mathcal{I}$, the statement follows since $\sigma_I$ is $\calF$-avoiding. 
If $r$ of the elements are in a common set $I$ but there is no other $r$-element subset of $P$ which is contained in an element of $\mathcal{I}$, then there is at most one sign in the sequence which is not $+$. This shows that in this case the $(r+1)$-subset is valid.  
For the remaining case, let $i\in J\setminus I$, $j\in I\setminus J$, without loss of generality $i<j$, and $I\cap J = A = \{a_1,\ldots,a_{r-1}\}$ for some $I, J \in \mathcal I$.
Since the elements $a_1,\ldots,a_{r-1}$ are consecutive in the ordered $(r+1)$-subset, there are only 3 options at what positions $i$ and $j$ can be, which is at the first two positions, the last two positions or the first and the last position. This also shows that it is impossible for an $r$-tuple other than $\{i\}\cup A$ and $\{j\}\cup A$ to be in a common family $K\in\mathcal I$. 
Otherwise all three pairs of missed elements $i,j$ and $k$ would have to choose a different one of these 3 options, which is~impossible.
%(, notably occupying the first two and the last two positions with just three elements).
 
 Hence there are only 3 cases left, depending on where our $r-1$ consecutive elements $A = I \cap J$ are in the $(r+1)$-element subset $P$:
\begin{align*}
    \sigma|_{(a_1,\ldots,a_{r-1},i,j)}=**{+}\ldots{+}, \
    \sigma|_{(i,a_1,\ldots,a_{r-1},j)}=*{+}\ldots{+}*,\
    \sigma|_{(i,j,a_1,\ldots,a_{r-1})}={+}\ldots{+}**,
\end{align*}
where $*$ denotes some value given by $\sigma_I$ or $\sigma_J$. 
Hence there are $r-1$ consecutive $+$ signs which shows that the partial sign mapping is $\calF$-avoiding.
\end{proof}

\begin{claim}
    For every model $M$ of $\phi$ there is a completion of $\sigma_\phi$.
\end{claim}

\begin{proof}
    For a model~$M$,
    we assign the corresponding values of the variables and literals to the triples in~$\sigma_\phi$.
    For each gadget in $\sigma_\phi$ on a subset $S\subset [N]$, there is a completion of $\sigma_S$. 
    Since any two gadgets intersect in at most
    a triple which encodes the sign of a variable or literal,
    the gadgets agree on their intersection.
    Furthermore, the needed consecutivity is fulfilled. 
    Propagator gadgets on 6 elements consist of smaller ones with a consecutive intersection, 
    so they can be completed using the combination lemma~\ref{lem:combination}. 
    Two propagator gadgets on 6 elements intersect either in the three elements corresponding to the common variable 
    or they intersect in at most three consecutive elements in the clause gadget. 
    Recall that the one variable does not occur more than once per clause.
    See Figure \ref{fig:intersect} for an illustration of the intersections.
    Thus we can apply the combination lemma~\ref{lem:combination} and obtain the desired completion of $\sigma_\phi$. 
    This completes the reduction.
\end{proof}

\begin{figure}[htb]
    \centering
    \includegraphics[page = 2,scale=0.85]{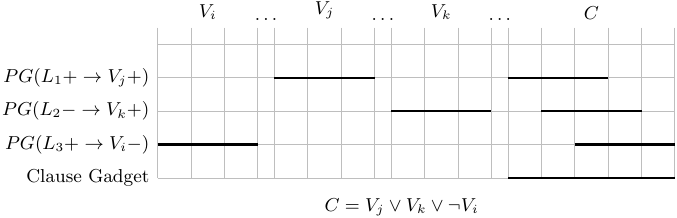}
    \caption{Intersection of propagator gadgets and clause gadgets. Clause gadgets do not intersect.}
    \label{fig:intersect}
\end{figure}

\section{Algorithm to find Gadgets}
\label{sec:algorithm}

%\manfred{TODO} \Cref{sec:prelim} referenzieren.

Our goal is to find gadgets for a reduction from 3SAT to a specific completion problem. The used notation is introduced in \Cref{sec:prelim}.
%\noindent
The gadget-finding problem can be expressed as
\begin{align*}  
    \exists \sigma : \calD \to \{+,?,-\} \ 
    \forall f:\calX \to \{+,-\}\ 
    [
    &(
    \neg\phi(f) \rightarrow 
    \forall  \text{ $\calF$-avoiding } \sigma^*  : \sigma_f \not\subset \sigma^*
    ) 
    \ \wedge \\
    &
    \phantom{\neg}(\phi(f) \rightarrow 
    \exists  \text{ $\calF$-avoiding } \sigma^*  : \sigma_f \subset \sigma^*
    )
    ]
\end{align*}
which by introducing indicator variables to encode the mappings $\sigma$, $f$, and $\sigma^*$, yields a quantified Boolean formula 
from the third level $\Sigma_3$ of the polynomial-time hierarchy.
For example, a propagator gadget $\sigma$ modelling $\psi = A \rightarrow B$ where the value of $A$ is encoded in the triple $(1,2,3)$ and the value of $B$ is encoded in the triple $(4,5,6)$ propagates the value of $+$ from $I(A) = (1,2,3)$ to a $+$-value of $I(B) = (4,5,6)$ in all completions of $\sigma$. The only assignment of truth values to $A$ and $B$ where there is no completion is when $A = +$ and $B = -$.

\medskip

Inspired by the conflict-driven clause learning (CDCL) algorithm for solving SAT instances,
we developed a gadget-finding problem algorithm (see \Cref{algo}).
To deal with the outermost quantifier, 
we use a SAT solver to enumerate partial assignments~$\sigma$.
For the second quantifier,
we loop over all assignments $f:\calX \to \{+,-\}$.
Since for our applications, 
the gadgets will have 2 or 3 variables, respectively,
the blowup is at most $2^3=8$.
To cope with the third quantifier,
we again use a SAT solver to decide 
whether $\sigma$ is completable with respect to the assignment~$f$.

To find a gadget, we make use of a partial order on the solution space of partial sign mappings. 
We say a partial mapping $\sigma_1: \calD'_1 \to \{+,-\}$ is smaller than a mapping $\sigma_2: \calD_2' \to \{+,-\}$ (short: $\sigma_1 \prec \sigma_2$)
if $\calD'_1 \subset \calD'_2$ and $\sigma_2(d) = \sigma_1(d)$ for all $d \in \calD_1'$. 
The essential property we use is: 
If $\sigma : \calD' \to \{+,-\}$ 
with $\calD' \subset \calD$ has an $\calF$-avoiding completion, then all $\sigma' \prec \sigma$ have an $\calF$-avoiding completion.

The \emph{down-set} $D(\sigma)$ consists of all $\sigma'$ with $\sigma' \preceq \sigma$,
and 
the \emph{up-set} $U(\sigma)$ consists of all $\sigma'$ with $\sigma' \succeq \sigma$.
If $\sigma_f$ is not completable 
while $\psi(f)=\True$ 
(the gadget is \emph{too strict}),
we search for a minimal too strict partial assignment $\sigma' \preceq \sigma$ and excludes its up-set by adding a constraint to the first-level SAT encoding.
Similarly, if $\sigma_f$ is completable 
while $\psi(f)=\False$ 
(the gadget is \emph{too loose}),
we search for a maximal too loose partial assignment $\sigma' \succeq \sigma$ and exclude its down-set.
Whenever on the third level the mapping $\sigma$
is classified as too strict or too loose,
further clauses will be added to the first level,
which allows to prune the search space significantly.

%\SetKwComment{Comment}{/* }{ */}

\newcommand{\Break}{\State \textbf{break} }

\begin{algorithm}[tbp]
\caption{An algorithm for the gadget-finding problem}\label{algo}

\textbf{Input:} \\
A combinatorial substructure on a domain $\calD$ encoded via forbidden substructures $\calF$,
a~quantifier-free Boolean formula 
$\psi$ on variables~$\calX$, 
%$\phi : \{x_1\ldots,x_k\} \to \{+,-\}$, 
%domain $\calD$,
%forbidden structures $\calF$,
and a mapping $I : \calX \hookrightarrow \calD$ synchronizing the values of the variables with the signs of some elements in $\calD$. 
%$I \subset \calD$.
%$I = \{i_1,\ldots,i_k\} \subset \calD$.
\\
\textbf{Output:} \\
A gadget modelling $\psi$, i.e., a partial assignment $\sigma : \calD \to \{+,?,-\}$
such that for every assignment $f:\calX \to \{+,-\}$,
$\sigma_f$~has an $\calF$-avoiding completion if and only if $\psi(f)=\True$.
($\sigma_f$ is obtained from $\sigma$ 
by further setting $\sigma_f(I(x)) := f(x)$ for $x \in \calX$.)

\begin{algorithmic}[1]

\State 
Create a CNF and 
use a SAT solver to 
enumerate partial assignments \par\noindent $\sigma : \calD \to \{+,?,-\}$

\While{$\exists$ partial assignment $\sigma$}
    \State valid := \True
    
    \For{all assignments $f:\calX \to \{+,-\}$}
        %create CNF and 
        \State 
        Use SAT solver to
        test whether $\sigma_f$ has an $\calF$-avoiding completion

        \If{$\psi(f)=+$ and $\sigma_f$ not completable ("too strict")}
            \State 
            %Set determined values of $\sigma$ to $?$ to obtain
            Choose a minimal $\sigma' \preceq \sigma$ such that $\sigma'_f$ is not completable \par \hspace{1cm}(again tested via SAT)
            \State
            Add constraint to the CNF to exclude the upset $U(\sigma')$
            \State 
            valid := \False
            %\Break
        \EndIf
        
        \If{$\psi(f)=-$ and $\sigma_f$ completable ("too loose")}
            \State 
            %Iteratively try to set $?$ of $\sigma$ to $+$ or $-$ to obtain
            Choose a maximal $\sigma' \succeq \sigma$ such that $\sigma'_f$ is completable 
            \par \hspace{1cm}(again tested via SAT)
            \State
            Add constraint to the CNF to exclude the downset $D(\sigma')$
            \State 
            valid := \False
            %\Break
        \EndIf
        
    \EndFor

    \If{valid} %$\forall f$: $\sigma_f$ completable if and only if $\phi(f)=+$}
        \State \Return{ $\sigma$ }
    \EndIf
    
\EndWhile
%\Return{}  \Comment{Nothing was found}
\end{algorithmic}
\end{algorithm}

%\clearpage

In an early version, 
we have also tested a DPLL-like strategy 
where we only excluded the up-set (resp.\ down-set) whenever a partial assignment $\sigma$ is too strict (resp.\ too loose). 
Formally, instead of steps~7 and~12 one just sets $\sigma' := \sigma$.
In the following we will refer to this variant as the \emph{basic algorithm}, while the algorithm presented in \Cref{algo} is referred to as \emph{advanced algorithm}.
Apparently,
searching for a minimal (resp.\ maximal) $\sigma'$ and excluding its up-set (resp.\ down-set)
significantly reduces the computing times, 
which is a similar behaviour as when going from DPLL to CDCL (cf.\ \cite[Chapter~4]{HandbookSatisfiablity2009}).\footnote{DPLL excludes the up-set of a infeasible partial assignment to solve a CNF instance. CDCL brings a significant speedup by searching a minimal infeasible partial assignment and excluding its up-set.}

\subsection{Computational Aspects}
\label{sec:computional_aspects}

Since the number of SAT instances created and solved by our algorithm is proportional to the number of prunings, an efficient pruning strategy in general yields better computation times.
We observed from our experiments that
the advanced algorithm (CDCL-like) prunes the search space 
much more efficient than the basic algorithm (DPLL-like), i.e., 
it comes to a conclusion faster and with fewer prunings.

Concerning the size of the gadgets, i.e., the number of underlying elements contained in this gadget,
our experiments showed that $n=r+2$ (which is the first nontrivial) is already a good choice for finding gadgets. 
In rank $r=3$,
a hardness reduction with gadgets of size $n=5$ is possible for 31 settings; cf.\ \Cref{fig:stat5}.
For the remaining 10 of the 41 hard settings (cf.\ \Cref{thm:rank3}), 
we found gadgets of size~$n=6$;  cf.\ \Cref{fig:stat6}.
Note that only 39 of the 41 hard settings were found within a timeout-limit of 5 CPU minutes per gadget-search.
The remaining 103 settings in rank 3 do not omit suitable gadgets of size 6 for a hardness reduction, but there might be gadgets of size 7 or larger.
In rank $r=4$, we only searched for gadgets of size $n=r+2=6$ with the advanced algorithm and a timeout of 100 seconds per gadget-search. Within 251 CPU hours our framework found 2328 hard settings among our \NUMSETTINGSFOURSELECTION benchmark settings (cf.\ \Cref{thm:rank4}). In order to reduce the search space of millions of settings, we investigated the benchmark consisting of the settings with no 2 instead of no 3 consecutive + signs in the forbidden family $\calF$.

\begin{figure}[tbp]
\begin{subfigure}{0.48\textwidth}
    \includegraphics[width=\textwidth]{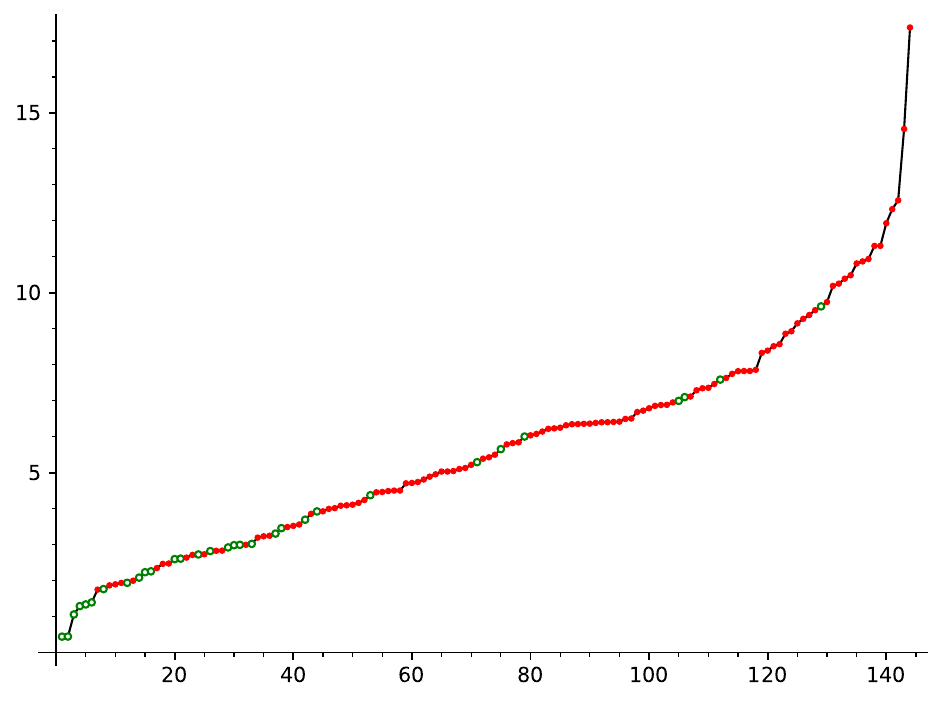}
    \caption{computing times of the basic algorithm}
    \label{fig:n5B_time}
\end{subfigure}
\hfill
\begin{subfigure}{0.48\textwidth}
    \includegraphics[width=\textwidth]{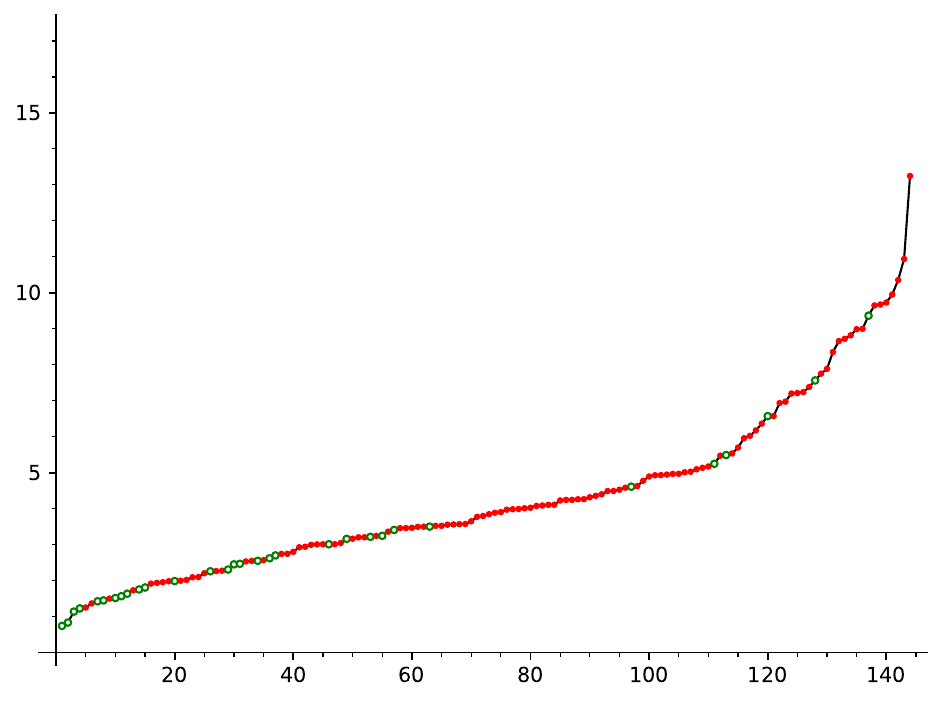}
    \caption{computing times of the advanced algorithm}
    \label{fig:n5A_time}
\end{subfigure}

\medskip

\begin{subfigure}{0.48\textwidth}
    \includegraphics[width=\textwidth]{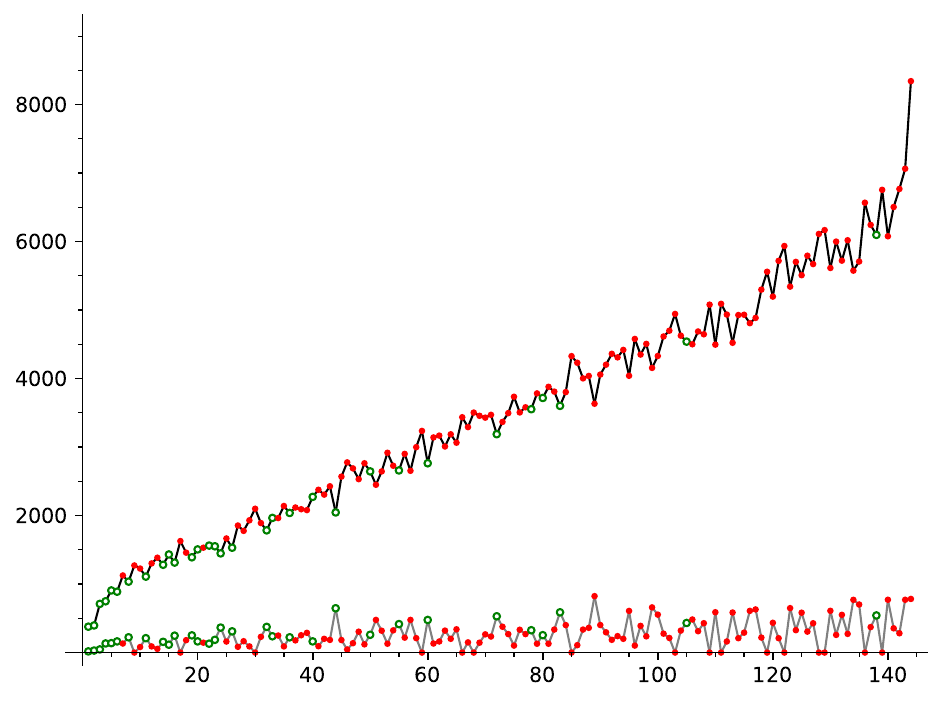}
    \caption{prunings of the basic algorithm}
    \label{fig:n5B_blacklist}
\end{subfigure}
\hfill
\begin{subfigure}{0.48\textwidth}
    \includegraphics[width=\textwidth]{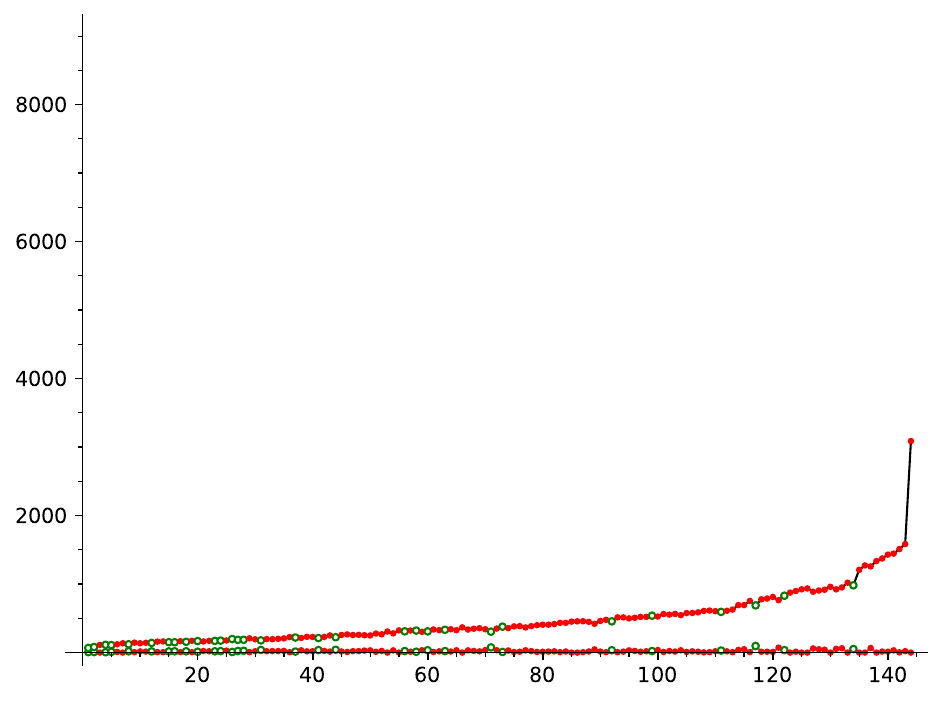}
    \caption{prunings of the advanced algorithm}
    \label{fig:n5A_blacklist}
\end{subfigure}

\caption{
Statistics when searching gadgets of size $n=5$ 
for the 144 settings in rank $r=3$
using the basic algorithm  
and 
the advanced algorithm, respectively.
After a total of 13.7 (resp.\ 10.3) CPU minutes,
31 settings were successfully certified NP-hard,
and the remaining 113 settings were proven not to contain the desired gadgets by the
basic algorithm  (resp.\ the advanced algorithm).
\subref{fig:n5B_time} and \subref{fig:n5A_time} show
the computing times in CPU seconds.
\subref{fig:n5B_blacklist} and \subref{fig:n5A_blacklist} show the number of blacklisting-events;
the gray/black curve shows the number of down/up-prunings. 
The certified/failed settings are marked with a green circle/red square.
%The X settings which were successfully certified NP-hard are marked with a circle;
%the Y settings for which at least one gadget-search timed out %after 300 seconds 
%are marked by a cross;
%the Z settings which do not omit suitable gadgets are marked with a square.
}
\label{fig:stat5}
\end{figure}

%\clearpage

\begin{figure}[tbp]
\begin{subfigure}{0.48\textwidth}
    \includegraphics[width=\textwidth]{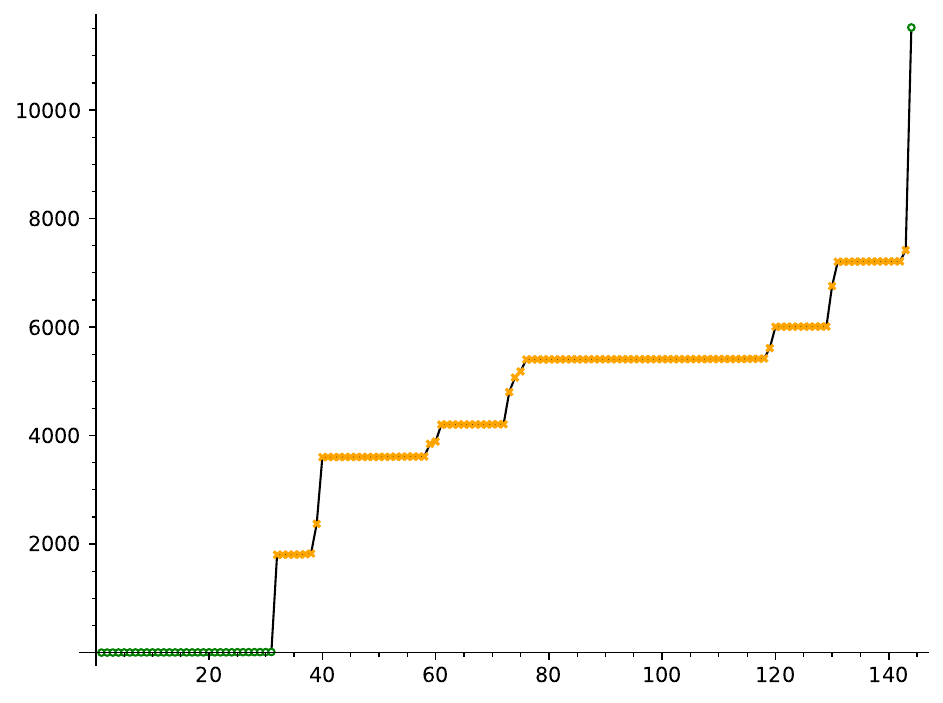}
    \caption{computing times of the basic algorithm}
    \label{fig:n6B_time}
\end{subfigure}
\hfill
\begin{subfigure}{0.48\textwidth}
    \includegraphics[width=\textwidth]{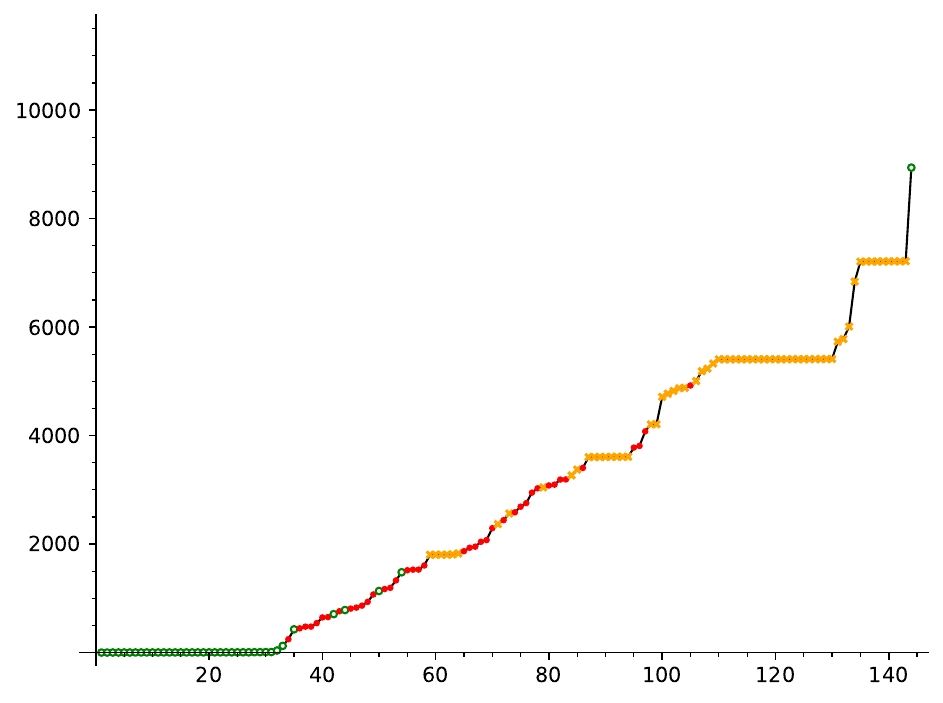}
    \caption{computing times of the advanced algorithm}
    \label{fig:n6A_time}
\end{subfigure}

\medskip

\begin{subfigure}{0.48\textwidth}
    \includegraphics[width=\textwidth]{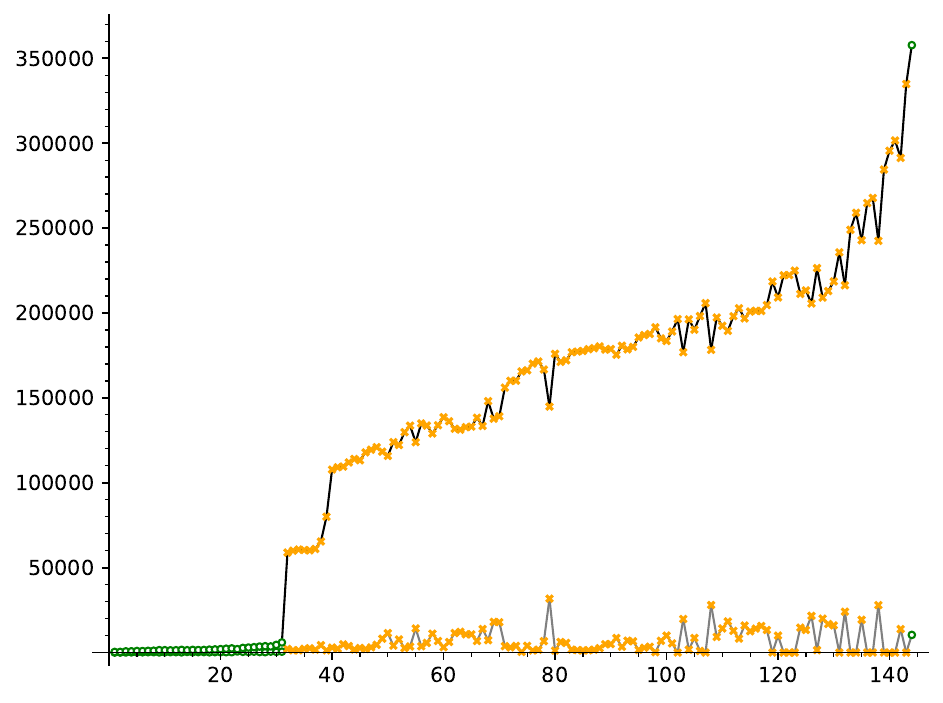}
    \caption{prunings of the basic algorithm}
    \label{fig:n6B_blacklist}
\end{subfigure}
\hfill
\begin{subfigure}{0.48\textwidth}
    \includegraphics[width=\textwidth]{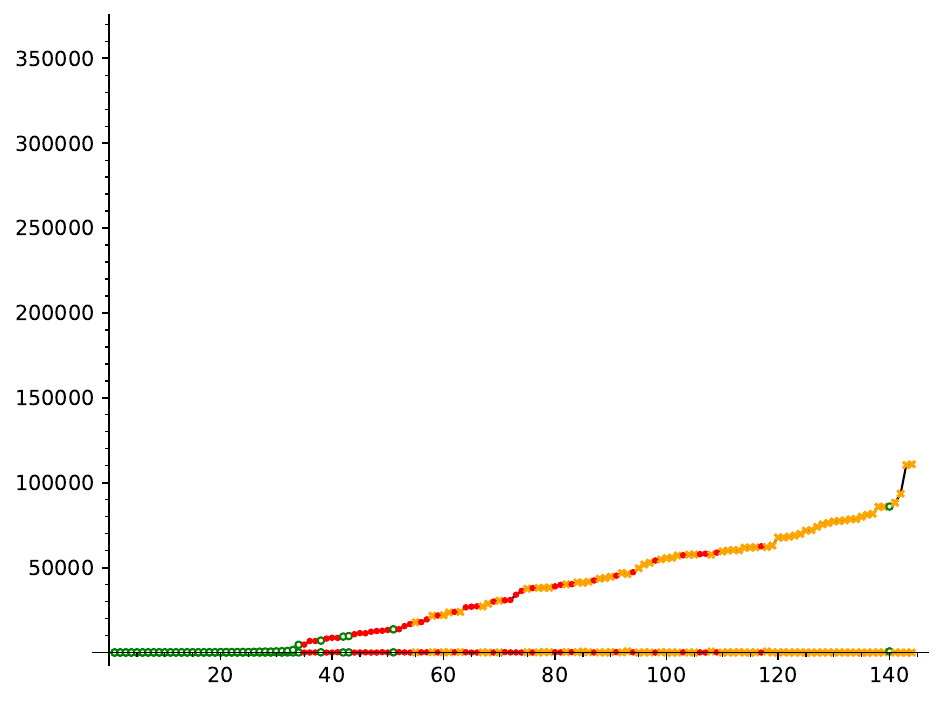}
    \caption{prunings of the advanced algorithm}
    \label{fig:n6A_blacklist}
\end{subfigure}

\caption{
Statistics when searching gadgets of size $n=6$ 
for the 144 settings in rank $r=3$
using 
%the basic algorithm  
%and 
%the advanced algorithm, respectively,
%and 
a timeout of 5 CPU minutes per gadget-search.
After a total of 158 (resp.\ 112) CPU hours,
32/112/0 (resp.\ 39/64/41) settings were successfully certified NP-hard/timed out/were proven not to contain the desired gadgets when running the basic algorithm (resp.\ advanced algorithm).
\subref{fig:n6B_time} and \subref{fig:n6A_time} show
the computing times in CPU seconds.
\subref{fig:n6B_blacklist} and \subref{fig:n6A_blacklist} show the number of blacklisting-events;
the gray/black curve shows the number of down/up-prunings. 
The certified/timed out/failed settings are marked with a green circle/orange cross/red square.
%The X settings which were successfully certified NP-hard are marked with a circle;
%the Y settings for which at least one gadget-search timed out %after 300 seconds 
%are marked by a cross;
%the Z settings which do not omit suitable gadgets are marked with a square.
}
\label{fig:stat6}
\end{figure}

As the search space grows significantly as the gadget size $n$ increases,
we first try to find gadgets of sizes $r+1, \ldots, n-1$
before we actually search for size~$n$ 
to keep the computing times as low as possible.
Since 31 of the hard settings in rank~$r=3$ have suitable  gadgets already for $n=5$,
some entries in statistics for $n=6$ come with little time-requirement; cf.\ the leftmost entries in \Cref{fig:n6B_time} and \Cref{fig:n6A_time}.

The CNF instances created internally are relatively small in the beginning but grow over time because clauses are added iteratively to prune the search space. 
For the enumeration, the initial CNF instances have 
$3 \binom{n}{r}$ variables and  $\binom{n}{r+1} \cdot |\calF| + O(n^r)$ clauses where $\calF$ denotes the set of forbidden patterns. 
The instances for testing completability are essentially the same, except that we pre-assign some of the variables by additional $O(n^r)$ clauses.
In the advanced algorithm, we create further instances for finding a minimal/maximal configuration which is too strict/too loose.
Recall that these will be utilized for a more efficient pruning.
To find such a minimal/maximal configuration $\sigma'$ for a configuration $\sigma$,
we iteratively search for configurations with one fewer/one additional element set.
To assert that the number of set entries is exactly $k$ for some integer $0 \le k \le \binom{n}{r}$, we use a sequential cardinality counter  in the CNF, which comes with $O(k n^r)$ auxiliary variables and constraints; cf.\ \cite{Sinz2005}.

\Cref{fig:stat5} and \Cref{fig:stat6}
show time and pruning statistics for our computations for gadgets of sizes $n=5$ and $n=6$,
and for both variants of the algorithm.
The main difference between the two variants of the algorithm is that the advanced algorithm (CDCL-like) has significantly less prunings than the basic algorithm (DPLL-like). 
This pruning in particular reduces the size of the instances and the computing times. 
On the other hand, the advanced algorithm has some overhead to compute an optimal pruning.
In practice, however,
the advanced algorithm already
outperformed the basic algorithm for $n=5$.
For $n=6$, the basic algorithm timed out for most of the settings.
This comparison clearly underlines the superiority of the advanced algorithm over the basic algorithm.

\section{Proof of Theorem~\ref{thm:rank3} for other settings}
\label{sec:others}

Only a specific combination of gadgets will work to conclude an NP-hardness proof as presented in Section~\ref{sec:proof}.
In the following we explain how we employed the search algorithm from Section~\ref{sec:algorithm} to eventually find the desired propagator gadgets and clause gadgets.

Recall that there are four possible disjunctions on two Boolean variables $x$ and~$y$:
\goodbreak
\begin{itemize}
\item 
$x \vee y$, which is logically equivalent to $\neg x \rightarrow y$ and  $\neg y \rightarrow x$;
\item 
$x \vee \neg y$, which is logically equivalent to $\neg x \rightarrow \neg y$ and  $y \rightarrow x$;
\item 
$\neg x \vee y$, which is logically equivalent to $x \rightarrow y$ and  $\neg y \rightarrow \neg x$;
\item 
$\neg x \vee \neg y$, which is logically equivalent to $x \rightarrow \neg y$ and  $y \rightarrow \neg x$;
\end{itemize}
\goodbreak
Hence we are interested in the following four possible propagator gadgets
%\begin{multicols}{2}
\begin{itemize}
\item 
$PG(X_1{-} \rightarrow  X_2{+}) = PG(X_2{-} \rightarrow  X_1{+})$
\item 
$PG(X_1{-} \rightarrow   X_2{-}) = PG(X_2{+} \rightarrow   X_1{+})$
\item 
$PG(X_1{+} \rightarrow  X_2{+}) = PG(X_2{-} \rightarrow  X_2{-})$
\item 
$PG(X_1{{+}} \rightarrow  X_2{-})=PG(X_2{{+}} \rightarrow  X_1{-})$
\end{itemize}
%\end{multicols}
\noindent
where $X_1=(a,a+1,a+2)$,  $X_2=(b,b+1,b+2)$,
with $a<b$. \\

\goodbreak
Similarly, there are eight candidates for a clause gadget:
\begin{multicols}{2}
\begin{itemize}
\item 
$CG(X_1{+} \vee X_2 {+}\vee X_3{+})$
\item 
$CG(X_1{+} \vee X_2{+} \vee  X_3{-})$
\item 
$CG(X_1{+}  \vee  X_2{-} \vee X_3{+})$
\item 
$CG(X_1{+}  \vee  X_2{-} \vee X_3-)$
\item 
$CG( X_1{-} \vee X_2{+} \vee X_3{+})$
\item 
$CG( X_1{-} \vee X_2{+} \vee  X_3-)$
\item 
$CG( X_1{-} \vee X_2{-} \vee X_3{+})$
\item 
$CG( X_1{-} \vee  X_2{-} \vee X_3{-})$
\end{itemize}
\end{multicols}

While in Section~\ref{sec:proof}
we used four propagator gadgets for the reduction,
it is also possible to continue with fewer.

\begin{itemize}
\item 
(Scenario 1:) If all four of the propagator gadgets exist,
it is sufficient to have one of the eight clause gadgets
to perform a hardness reduction analogous to Section~\ref{sec:proof}.
Note that it is sufficient to find the two propagator gadgets 
$PG(X_2{-} \rightarrow X_1{+})$ and $PG(X_2{+} \rightarrow X_1{-})$, 
because we  build the other two by combining those two. 
To build a propagator gadget 
$PG(X_3{-} \rightarrow X_1{-})$, 
we use
$PG(X_3{-} \rightarrow X_2{+})$ and $PG(X_2{+} \rightarrow X_1{-})$
for $X_1=(a,a+1,a+1)$,  $X_2=(b,b+1,b+2)$, $X_3=(c,c+1,c+2)$ with $a<b<c$. 
Note that this encoding is slightly bigger (by  a factor of at most two)
because the negation of a variable is stored in an auxiliary triple.
Building a 
$PG(X_3{+} \rightarrow X_1{+})$ gadget is~analogous. 
\iffalse
\manfred{alte version von felix:}
More specifically,
to do so, we need slightly more space for the encoding, 
because we want to store the negation of a variables twice, 
connected by a $PG(X_2{-} \rightarrow X_1{+})$ and a 
$PG(X_2{+} \rightarrow X_1{-})$ gadget respectively. Now instead of using a 
$PG(X_2{+} \rightarrow X_1{+})$ gadget, we connect the clause with a 
$PG(X_2{+} \rightarrow X_1{-})$ gadget to the encoded negation that is connected to
the actual variable with the $PG(X_2{-} \rightarrow X_1{+})$ gadget. Building a 
$PG(X_2{-} \rightarrow X_1{-})$ gadget is analogous. 
\fi
\end{itemize}
Otherwise, if those two propagator gadgets do not exist,
we need to be careful.
Logical values need to be propagated either to the left or to the right side
and therefore the choice of the clause gadget and 
the positive/negated variables in the clause 
play a central role.

If there exists a clause gadget for $c = x_1 \vee x_2 \vee x_3$, that is, the disjunction of three positive~variables,
we proceed as following:
\begin{itemize}
\item 
(Scenario 2:) 
If $PG(X_2{+} \rightarrow X_1{+})$ and $PG(X_2{+} \rightarrow X_1{-})$ exist, we  proceed as in Section~\ref{sec:proof} and place the elements for the variables to the left of the elements for the~clauses.
\item 
(Scenario 3:) 
If $PG(X_1{+} \rightarrow X_2{+})$ and $PG(X_1{+} \rightarrow X_2{-})$ exist, we  proceed as in Section~\ref{sec:proof} but place the elements for the variables to the right of the elements for the~clauses.
\end{itemize}
Analogously, we deal with the clause gadget for clauses on three negative variables $c = \neg x_1 \vee \neg x_2 \vee \neg x_3$:
\begin{itemize}
\item 
(Scenario 4:) 
If $PG(X_2{-} \rightarrow X_1{+})$ and $PG(X_2{-} \rightarrow X_1{-})$ exist, we  proceed as in Section~\ref{sec:proof} and place the elements for the variables to the left of the elements for the~clauses.
\item 
(Scenario 5:) 
If $PG(X_1{-} \rightarrow X_2{+})$ and $PG(X_1{-} \rightarrow X_2{-})$ exist, we  proceed as in Section~\ref{sec:proof} but place the elements for the variables to the right of the elements for the~clauses.
\end{itemize}

We have implemented precisely those five scenarios in our framework.
All \NUMPROOFS settings were found and verified by running these tests.
To exclude errors from the SAT solver,
we check the correctness of any model returned by the solver 
in the case a CNF is satisfiable.
Otherwise, if the CNF is unsatisfiable,
the solver can output a DRAT certificate and we employ the independent proof-checking tool DRAT-trim~\cite{WetzlerHeuleHunt2014} to verify the certificate.
Technical details and instructions for verifying and reproducing our experiments are deferred to the README file of our supplemental data \cite{supplemental_data}.

\section{Proof of Theorem~\ref{thm:even_rank}}
\label{app:even_rank}

Let $r \geq 4$ be an even integer.
We construct gadgets which fulfill the requirements of Scenario~1, that is, all four propagator gadgets exists;  cf.\ \Cref{sec:others}. 
As described above it is sufficient to construct the two propagator gadgets  $PG(X_2{-} \to X_1{+})$ and $PG(X_2{+} \to X_1{-})$
as the other two can be constructed from them. 
With all four propagator gadgets at hand, 
any clause gadget is sufficient 
to complete the reduction. 
%In the following we construct the two propagator gadgets and 
We will construct
the clause gadget $CG(X_1{+} \vee X_2{+} \vee X_3{+})$.

\subsection{Propagator Gadget \texorpdfstring{$PG(X_2{-} \to X_1{+})$}{PG(X2- -> X1+)}}
We construct a gadget $PG(X_2{-} \to X_1{+})$ on $r+1$ elements $a_1 < \ldots < a_{r+1}$ such that the sign of the $r$-tuple $X_1 = (a_1, \ldots, a_r)$ encodes the value of variable $x_1$ and $X_2 = (a_2, \ldots, a_{r+1})$ encodes~$x_2$.
Let $\sigma_P$ be the mapping restricted to the considered $r+1$ elements with the following signs
$\sigma_P(a_1, \ldots a_{i-1}, a_{i+1}, \ldots a_{r+1}) = (-)^{i}$ for $i = 2, \ldots, r$.
This gives the following sign sequence on 
%all $r+1$ elements:
the elements
$\{a_1,\ldots,a_{r+1}\}$
\begin{align*}
    ? \underbrace{{+}{-}{+}{-}\ldots{+}}_{\text{$r-1$ signs}}?,
\end{align*}
where the first $?$ symbol indicates that $\sigma_P(X_1)$ is free and the last $?$ indicates that $\sigma_P(X_2)$ is free.
Hence for the propagator gadget we have the following property.

\begin{claim}
    Every completion $\sigma$ of $\sigma_P$ with $\sigma(X_2) = - $ fulfills $\sigma(X_1) = +$. 
    Moreover, if we additionally assume $\sigma_P(X_2) = + $ or $\sigma_P(X_1) = +$ there exists a completion. 
\end{claim}

\subsection{Propagator Gadget \texorpdfstring{$PG(X_2+ \to X_1-)$}{PG(X2+ -> X1-)}}

Similar as in the previous gadget, we use $X_1 = (a_1, \ldots, a_r)$ and $X_2 = (a_2, \ldots, a_{r+1})$ to encode $x_1$ and $x_2$, respectively. 
The mapping $\sigma_P$ restricted to the $r+1$ elements 
$a_1 < \ldots < a_{r+1}$ has the signs
$\sigma_P(a_1, \ldots a_{i-1}, a_{i+1}, \ldots a_{r+1}) = (-)^{i+1}$ for $i = 2, \ldots, r$.
This gives the following sign sequence on 
the elements
$\{a_1,\ldots,a_{r+1}\}$
\begin{align*}
    ? \underbrace{{-}{+}{-}{+}\ldots{-}}_{\text{$r-1$ signs}}?,
\end{align*}
where the first $?$ symbol indicates that $\sigma_P(X_1)$ is free and the last $?$ indicates that $\sigma_P(X_2)$ is free.
Hence for the propagator gadget we have the following property

\begin{claim}
    Every completion $\sigma$ of $\sigma_P$ with $\sigma(X_2) = + $ fulfills $\sigma(X_1) = -$. 
    Moreover, if we additionally assume $\sigma_P(X_2) = - $ or $\sigma_P(X_1) = -$ there exists a completion. 
\end{claim}

\subsection{Clause Gadget \texorpdfstring{$CG(X_1{+} \vee X_2{+} \vee X_3{+})$}{CG(X1+ v X2+ v X3+)}}

We construct a clause gadget on $r+2$ consecutive elements $C = \{c_1, \ldots, c_{r+2}\}$
with $X_1 = (c_1, \ldots, c_r)$, $X_2 =(c_2, \ldots, c_{r+1} ) $ , $X_3 = (c_3, \ldots, c_{r+2}) $ with the following signs:
\begin{align}
    \sigma_C(c_1, c_2, \ldots c_{i-1}, c_{i+1}, \ldots, c_{r},c_{r+2}) &= (-)^{i}  &\text{for } i = 2, \ldots, r \label{rule:ir+1}\\ 
    \sigma_C(c_2, c_3, \ldots, c_{i-1}, c_{i+1} , \ldots, c_{r+1}, c_{r+2}) &= (-)^{i-1} &\text{for } i = 3, \ldots, r \label{rule:1i}\\
    \sigma_C(c_1, c_2, \ldots, c_{i-1}, c_{i+1}, \ldots c_{r+1}) &= (-)^i  &\text{for } i = 2, \ldots, r-1\label{rule:il} \\
    \sigma_C(c_1, \ldots, c_{r-1}, c_{r+1}) &= -\label{rule:rl}
\end{align}

Note that those four rules for the signs do not overlap, i.e. every $r$-tuple gets a sign from at most one of the rules. 
Moreover, the so constructed partial sign mapping is $\calF$-avoiding and has the following property. 
\begin{claim}
    Every completion $\sigma$ of $\sigma_C$ fulfills $\sigma(X_1) = {+} \vee \sigma(X_2) = {+} \vee \sigma(X_3) = {+}$. 
    Moreover, 
    if we assign $\sigma_C(X_1) = {+}$, $\sigma_C(X_2) = +$, or $\sigma_C(X_3) = {+}$, then
    there exists a completion $\sigma$.
\end{claim}

\begin{proof}
To show the first part, suppose towards a contradiction that there is a completion $\sigma$ which does not fulfill $\sigma(X_1) = {+} \vee \sigma(X_2) = {+} \vee \sigma(X_3) = {+}$. Hence it holds $\sigma(X_1) = {-}$, $\sigma(X_2) = {-}$ and $\sigma(X_3) = {-}$. 
By rule (\ref{rule:1i}),
the sign sequence of the $(r+1)$-element subset $\{c_2, \ldots, c_{r+1}, c_{r+2}\}$ is
\begin{align*}
    \sigma(X_2) \sigma(Y) -+ \ldots -+ \sigma(X_3),
\end{align*}
where $Y:=(c_2,\ldots,c_r,c_{r+2})$. 
Note that the sign of $Y$ 
is not determined by the four rules.
If $\sigma(X_2) = \sigma(X_3) =-$, it implies $\sigma(Y) = -$ in order to avoid the alternating sign pattern.

Next observe that, by rule (\ref{rule:ir+1}),
the sequence induced by the $(r+1)$-element subset $\{c_1, \ldots, c_{r}, c_{r+2}\}$
is
\begin{align*}
    \sigma(X_1) +- \ldots + \sigma(Y).
\end{align*}
By the previous observation we know that $\sigma(Y) = -$, a contradiction since 
$ \sigma(X_1) = -$. 

To show the second part, we consider the gadget and additionally 
assume that $\sigma_C(X_1) = {+}$, $\sigma_C(X_2) = +$, or $\sigma_C(X_3) = {+}$ is assigned. 
Clearly in those cases the two sign sequences above are valid. 
It remains to check the other sign sequences. 
As we see all of them have two consecutive signs which are the same and hence they are $\calF$-avoiding.  
Let us first have a look at the sign sequence of the $\{r+1\}$-subset $\{c_1, \ldots, c_{r+1}\} $. 
By rules (\ref{rule:il}) and (\ref{rule:rl}),
the sign sequence is 
\begin{align*}
    \sigma(X_1) - - + \ldots -+ \sigma(X_2),
\end{align*}
which is valid independent from the sign of $ \sigma(X_1)$ and $\sigma(X_2)$. 

To complete the proof, it remains to check the sign sequence of the $(r+1)$-element subset $\{c_1, \ldots, c_{j-1}, c_{j+1}, \ldots,  c_{r+2}\} $
for all $j = 2, \ldots, r$.
The signs of the two $r$-tuples $(c_1, \ldots, \allowbreak c_{j-1}, c_{j+1}, \ldots, \allowbreak c_r, c_{r+2}) $ and $(c_1, \ldots, c_{j-1}, c_{j+1}, \ldots,c_r, c_{r+1}) $ are consecutive.
By rules (\ref{rule:ir+1}) and (\ref{rule:il}), these two signs are equal  
\begin{align*}
    \sigma(c_1, \ldots, c_{j-1}, c_{j+1}, \ldots,c_r, c_{r+2})  = (-)^j = 
    \sigma(c_1, \ldots, c_{j-1}, c_{j+1}, \ldots,c_r, c_{r+1}),
\end{align*}
and hence also this sign sequence cannot be alternating.
This completes the proof.
\end{proof}

\section{Conclusion}
\label{sec:discussion}

We presented a SAT-based approach for finding gadgets for hardness reductions for the completion problem of sign mappings in an automated manner.
Our framework found  \NP-hardness reductions for thousands of structures and, inspired by the computational data, we were able to construct an infinite family for which \Complete is \NP-hard. 

Even though our framework is able to find gadgets automatically in reasonable time, 
the challenging task remains to find a suitable combination lemma for the reduction.

Our framework might be of independent interest
since it can be adapted to work with any related combinatorial structure and logical puzzles.
Moreover, it may help to find gadgets for reductions for other complexity classes such as \APX{}.

%%
%% Bibliography
%%

{
	\small
	\bibliographystyle{alphaabbrv-url}
	\bibliography{references}
}

\end{document}